\documentclass[11pt,leqno,fleqn]{article}

\newcommand{\Comment}[1]{\relax}
\usepackage{comment}
\usepackage{caption}
\usepackage{subcaption}
\usepackage{amsthm}
\usepackage{latexsym}
\usepackage{amsmath}
\usepackage{amssymb}
\usepackage{amsfonts}
\usepackage{mathtools}
\usepackage{algorithm}
\usepackage{algorithmicx}
\usepackage{amsbsy}
\usepackage{synttree} % to draw parse trees
\usepackage{diagrams}
\usepackage{enumitem} %  customize \begin{enumerate} ... \end{enumerate}
\usepackage[usenames]{color}
\usepackage[colorlinks=false, %
           citecolor=red,filecolor=blue,linkcolor=blue,urlcolor=blue,
           linktoc=page,
           pagebackref=true
           ]{hyperref} % for hyperlinks
\usepackage{fullpage} % simply by virtue of including this package,
\usepackage{afterpage}  % used in the command \afterpage{\clearpage}.
\usepackage{float} % for option [H] in float placement,
\usepackage{wrapfig}  % PACKAGE wrapfig MUST COME RIGHT AFTER PACKAGE float
\usepackage{times}
\usepackage{bm}
\usepackage{stmaryrd} % required for \llbracket and \rrbracket
\usepackage{MnSymbol}  % for \llangle and \rrangle
\usepackage{mathrsfs}  % required for \mathscr (``script''), to
\usepackage{graphicx}
\usepackage[all]{xy}
\usepackage{fancybox}
\usepackage{alltt}

\newcommand{\Hide}[1]{}

\newif
\ifnote
\notetrue
\newif
\ifTR
\TRfalse

\newtheorem{theorem}{Theorem}
\numberwithin{theorem}{section}
\newtheorem{lemma}[theorem]{Lemma}
\newtheorem{corollary}[theorem]{Corollary}

\newtheorem{fact}[theorem]{Fact}
\newtheorem{propxxx}[theorem]{Property}
\newenvironment{property}{\begin{propxxx}\rm}{\hfill\QED\end{propxxx}}
\newtheorem{restxxx}[theorem]{Restriction}

\newtheorem{agreexxx}[theorem]{Agreement}

\newtheorem{termxxx}[theorem]{Terminology}

\newtheorem{notxxx}[theorem]{Notation}

\newtheorem{assumxxx}[theorem]{Assumption}

\newtheorem{exaxxx}[theorem]{Example}
\newenvironment{example}{\begin{exaxxx}\rm}{\hfill\QED\end{exaxxx}}
\newtheorem{remxxx}[theorem]{Remark}

\newtheorem{openxxx}[theorem]{Open Problem}
\newtheorem{conjxxx}[theorem]{Conjecture}

\newtheorem{defxxx}[theorem]{Definition}
\newenvironment{definition}[1]{\begin{defxxx}[\emph{#1}]\rm}%
{\hfill\QED\end{defxxx}}
\newtheorem{procxxx}[theorem]{Procedure}
{\hfill\QED\end{procxxx}}

\newtheorem{Prxxx}[theorem]{Proof}
{\end{Prxxx}} % {\hfill\QED\end{Prxxx}}

  {\addtolength{\leftskip}{#1}\addtolength{\rightskip}{#2}}{\par}
\newcommand{\Set}[1]{\{ #1 \}}

\newcommand{\Let}[3]%
    {\textbf{\textsf{let}}\ {#1}\,{#2}\ \textbf{\textsf{in}}\;{#3}\,}
\newcommand{\Try}[3]%
    {\textbf{\textsf{try}}\ {#1} {#2}\ \textbf{\textsf{in}}\;{#3}\;}
\newcommand{\Mix}[3]%
    {\textbf{\textsf{mix}}\ {#1} {#2}\ \textbf{\textsf{in}}\;{#3}\;}
\newcommand{\LET}[3]%
    {\textbf{\textsf{let}}^*\ {#1} {#2}\ \textbf{\textsf{in}}\;{#3}\;}
\newcommand{\Letrec}[3]%
    {\textbf{\textsf{letrec}}\ {#1} {#2}\ \textbf{\textsf{in}}\;{#3}\;}

\newcommand{\QED}{{\Large $\square$}} 
\newcommand{\vfi}{\varphi}

\newcommand{\spacing}[2]{
  \renewcommand{\baselinestretch}{#2}
  \small\normalsize #1
  \setlength{\parskip}{0.1\baselineskip}
  \settowidth{\parindent}{xxxx}
  \setlength{\parindent}{#2\parindent}
  \setlength{\leftmargini}{\parindent}
  \setlength{\leftmarginii}{\parindent}
  \setlength{\leftmarginiii}{\parindent}
  \setlength{\footnotesep}{#2\footnotesep}
}

\begin{document}

\spacing{\normalsize}{0.98}
\setcounter{page}{1}     % \thispagestyle{empty} has to be inserted right
                         % after \maketitle to suppress the page numbering
                         % on the title page
\setcounter{tocdepth}{1} % to suppress subsections and subsubsections
                         % in the table of contents.
\ifTR
  \pagenumbering{roman} % for title page and table-of-contents page
\else
\fi

\title{Linear Arrangement of Halin Graphs}

\author{ Saber Mirzaei \\
       Boston University  \\
        \ifTR Boston, Massachusetts \\
        \href{mailto:smirzaei@bu.edu}{smirzaei{@}bu.edu}
        \else \fi
\and
Assaf Kfoury \\
        Boston University \\
        \ifTR Boston, Massachusetts \\
        \href{mailto:kfoury@bu.edu}{kfoury{@}bu.edu}
        \else \fi
}

\ifTR
   \date{\today}
\else
   \date{} %
\fi
\maketitle
  \ifTR
     \thispagestyle{empty} % it has to be inserted right after \maketitle
                           % in order to suppress the page numbering
  \else
  \fi

\vspace{-.3in}
  \begin{abstract}
  We study the Optimal Linear Arrangement (OLA) problem of Halin
graphs, one of the simplest classes of non-outerplanar graphs.  We
present several properties of OLA of general Halin graphs. We prove
a lower bound on the cost of OLA of any Halin graph, and define
classes of Halin graphs for which the cost of OLA matches this lower
bound. We show for these classes of Halin graphs, OLA can be computed
in $O(n\log{n})$, where $n$ is the number of vertices. 
  \end{abstract}

\ifTR
    \newpage
    \tableofcontents
    \newpage
    \pagenumbering{arabic}
\else
    \vspace{-.2in}
\fi

\section{Introduction}
\label{sect:intro}
  Given graph $G=(V,E)$, a \emph{linear arrangement} or simply
a \emph{layout} of vertices is defined as a bijective function
$\vfi : V \rightarrow \Set{1, \dots, |V|}$.  In the Optimal Linear
Arrangement (OLA) problem, a special case of more general vertex
layout problems, the goal is to find the layout $\vfi$ minimizing
$\sum_{\Set{v,u} \in E} |\vfi(v) - \vfi(u)|$. The OLA problem is known to
be NP-hard for general graphs~\cite{GareyJS76}, for \emph{bipartite
graphs}~\cite{even1975np}, and for more specific classes of graphs
such as \emph{interval graphs} and
\emph{permutation graphs}~\cite{cohen2006optimal}.

Defining interesting classes of graphs for which the OLA problem
is polynomially solvable has been a notoriously hard task. The results
have been few and spread over several decades
\cite{GareyJS76, even1975np, cohen2006optimal, chung1984, rostami2008minimum}. Three decades ago,
it was suggested that a good candidate for polynomial-solvable OLA
are interval graphs, a class of graph for which no NP-hardness
results were known at that time (page 13 of~\cite{johnson1985np}).
Efforts in that direction were in vain, as some two decades
later, the OLA problem of interval graphs was shown to be NP-hard~\cite{cohen2006optimal}.
Another candidate for polynomial-solvable OLA are the so-called \emph{recursively constructed
graphs}~\cite{horton2003linear},
given that most NP-hard problems on general graphs are easily
solvable for this class.

Halin graphs are planar graphs which the degree of every vertex is at
least 3 and can be constructed using an underlying tree $T$ and a
cycle $C$ which connects leaf nodes of the tree
$T$. Figure~\ref{fig:halin} presents a Halin graph. Throughout this
report, the edges of cycle $C$ are presented in bold and the edges of
the tree $T$ are depicted in dashed lines.  Halin graphs can
be considered as one of the simplest class of graphs that are not
outerplanar \footnote{A Halin graph is a 2-outerplanar graph.}.

To the best of our knowledge, the OLA problem of Halin graphs is only
studied for the simple case where the underlying tree is a caterpillar
~\cite{easton1996solvable}. After introducing our notations and
preliminary definitions in Section~\ref{sect:prelim}, we present
several properties of OLA of Halin graphs in
Section~\ref{sect:OLA-halin-graphs}, including a lower bound on the
cost of OLA for Halin graphs.  In
Section~\ref{sect:halin-solvable-case}, we define and study a class of
Halin graphs for which the cost of their OLA meets this lower
bound. We also present an algorithm which, given a Halin graph in this
class, returns an OLA in $O(n\log{n})$ where $n$ is the number of vertices.

\begin{figure}[h]
    \begin{center}
        \includegraphics[scale=0.45]{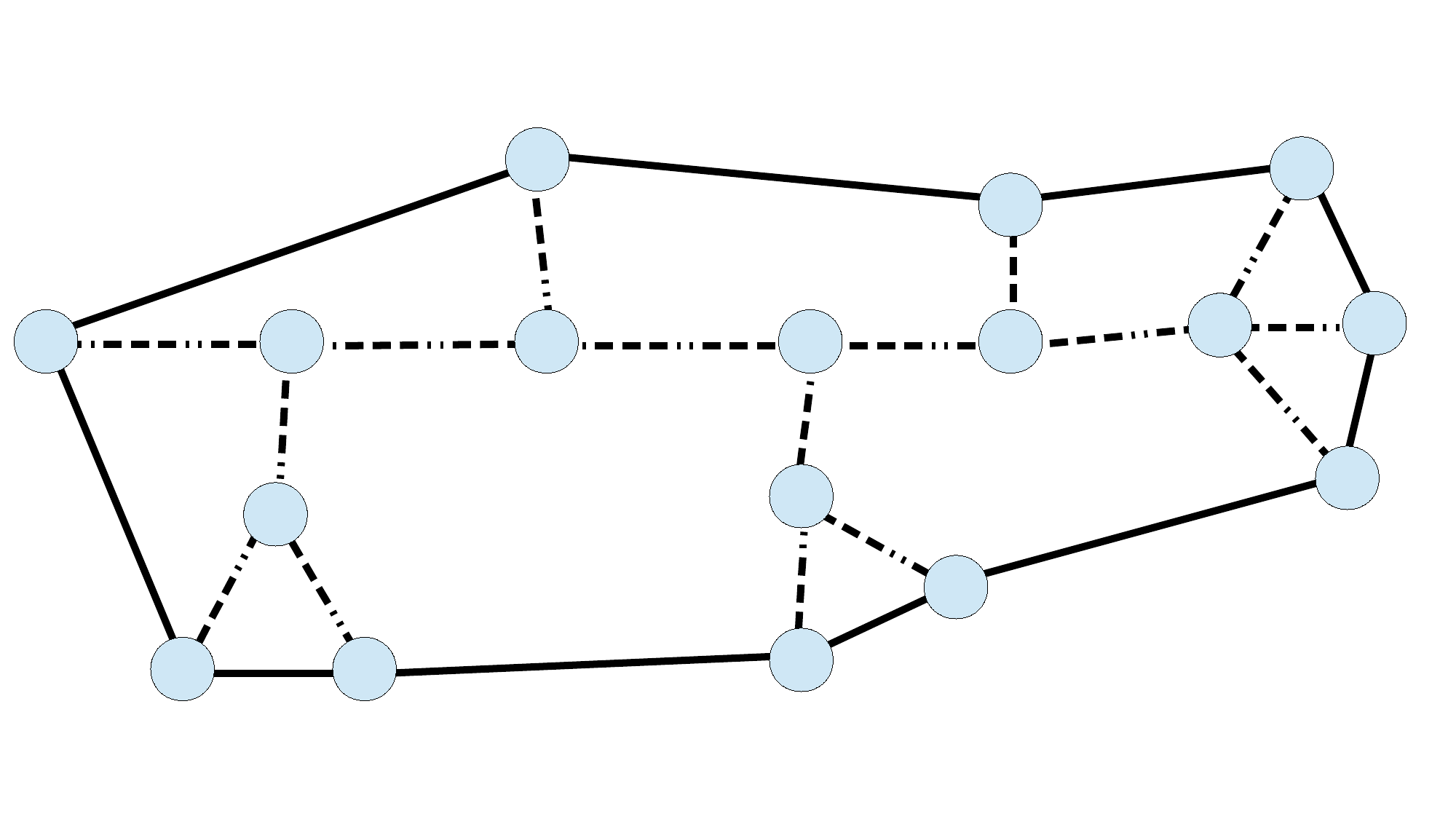}
    \end{center}
    \vspace{-0.2in}
    \caption{$H_1$, an example of a Halin graph.
    The cycle which connect the leaf nodes is shown in bold and
    the underlying tree is presented in dashed lines. }
    \label{fig:halin}
    \vspace{-0.2in}
\end{figure}

\section{Preliminaries}
\label{sect:prelim}
  We only consider simple, undirected graphs. For a finite graph $G=(V,E)$ where $V$ and $E$ are respectively the sets of vertices and
edges, we show $|V|$ by $n$ and $|E|$ as $m$. For a given vertex $v \in V$, $d_G(v)$ presents the degree of $v$ in $G$.
For a subgraph $G' \subseteq G$, $V(G')$ and $E(G')$ respectively present the set of vertices and edges of $G'$.

\noindent
We denote by $\Phi(G)$ the set of all possible layouts for the graph $G$.
A layout $\vfi$ can be considered as an ordering $(w_1, w_2, \ldots, w_n)$ of vertices of $V$. Accordingly
for $v = w_i \in V$, $\vfi(v) = i$. Without loss of generality we assume the left most and right most vertices are recursively labeled as
$1$ and $n$ and we call them the \emph{extreme vertices} based on $\vfi$.
\begin{notxxx}
Let $V_1,\ldots, V_k$ be a partitioning of $V$. We say a layout $\vfi$ is of type $(V_1,V_2,\ldots,V_k)$ if:
\[
    \forall 1 \le i < j \le k, \forall v \in V_i, \forall u \in V_j \Rightarrow \vfi(v) < \vfi(u)
\]
\end{notxxx}
\begin{notxxx}
Given layout $\vfi$ for $G=(V,E)$ and an edge $e = \Set{u,v} \in E$ we define the \emph{expand} of $e$ as:
\[
    \lambda(e, \vfi) = |\vfi(u) - \vfi(v)|
\]
\end{notxxx}
\noindent
Several cost functions have been defined on a given graph $G$ and layout $\vfi$. For a comprehensive list refer to~\cite{petit2013addenda}.
In this report we focus on $Optimal Linear Arrangement$ problem (OLA) defined as follows.
\begin{definition}{Optimal Linear Arrangement} Given an undirected graph $G=(V,E)$ and a layout $\vfi$
the \emph{linear arrangement} cost (LA) of $\vfi$ is:
\[
    LA(\vfi, G) = \sum_{\Set{v,u} \in E(G)} \lambda(\Set{u,v},\vfi)
\]
A layout $\vfi ^*$ is optimal if:
\[
     LA(\vfi ^*, G) = \min_{\vfi \in \Phi(G)} {LA(\vfi, G)}
\]
\end{definition}
Lemma~\ref{lemma:labeling-disjoint-edge} present a lower bound on the cost of optimal linear arrangement which
will be useful in presenting some properties and proofs in the rest of the paper.
\begin{lemma}
\label{lemma:labeling-disjoint-edge}
Given graph $G=(V,E)$ and two induced subgraph $G_1 = (V,E_1)$ and $G_2 = (V, E_2)$ s.t. $E_1 \cap E_2 = \emptyset$ and $E_1 \uplus E_2 = E$,
assume $\vfi ^*$, $\vfi ^*_{1}$ and $\vfi ^*_{2}$ are respectively the optimal linear arrangement for $G$, $G_1$ and $G_2$.Then $LA(\vfi, G) \ge LA(\vfi^ *_{1}, G_1) + LA(\vfi^ *_{2}, G_2)$.
\end{lemma}
A Halin graph $H=(V,E)$ is constructed based on an underlying tree $T = (V, E')$ embedded in plane in a planar manner where all the leaf
nodes are connected using a cycle $C=(V',E'')$. A Halin graph $H$ is shown as $H = T \uplus C$. As depicted in Figure~\ref{fig:non-unique-Halin}
it's easy to see that for a given tree
$T$, there may exist finitely many non-isomorphic Halin graphs.
\begin{figure}
\centering
        \begin{subfigure}[b]{0.3\textwidth}
                \includegraphics[scale=.5]
                {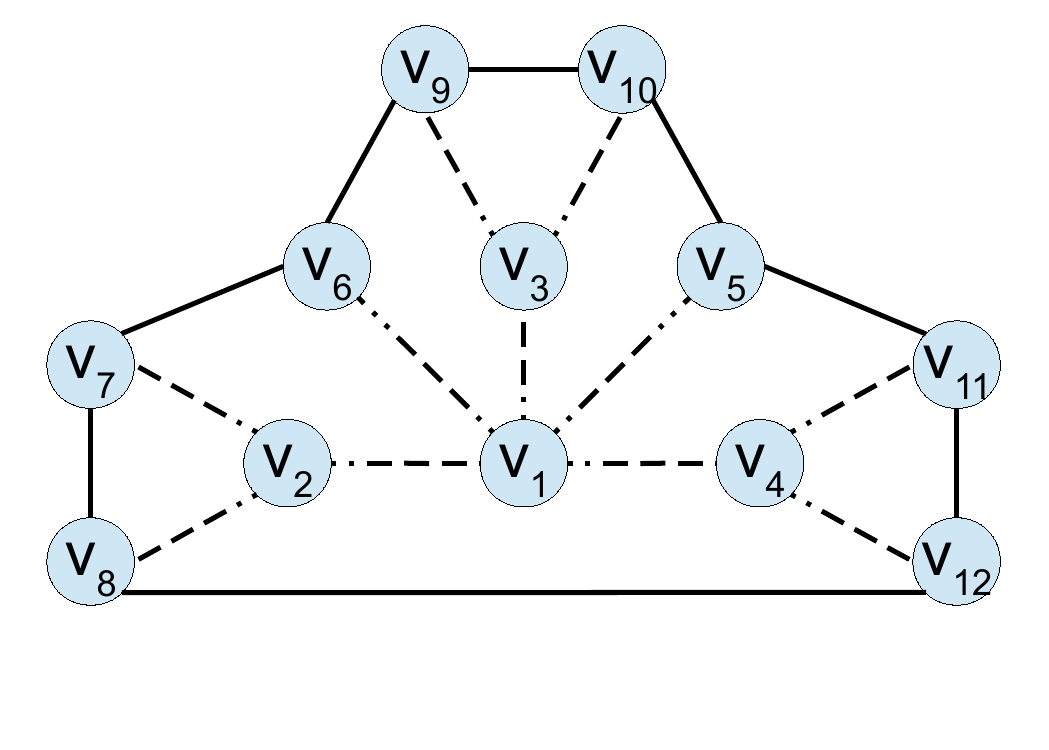}
                \caption{Halin graph $H_2$}
                \label{fig:non-RBT-exp-2}
        \end{subfigure}
        \qquad\qquad\qquad
        \begin{subfigure}[b]{0.3\textwidth}
                \includegraphics[scale=.5]
                {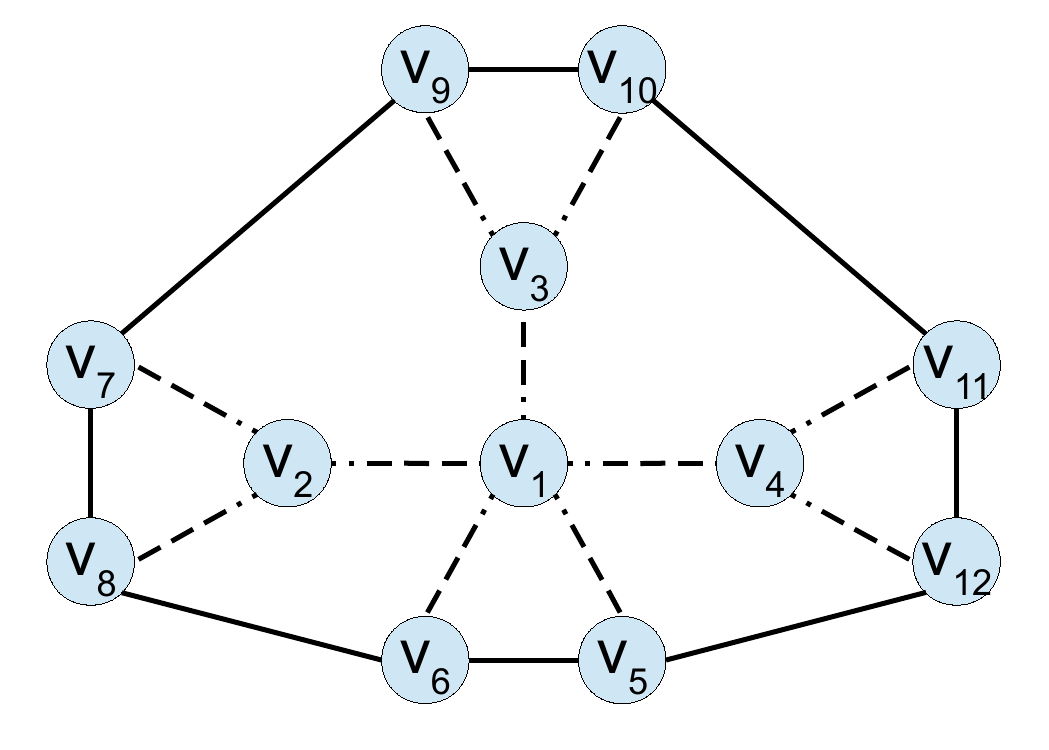}
                \caption{Halin graph $H_3$}
                \label{fig:non-RBT-exp-3}
        \end{subfigure}
        \caption{Two non-isomorphic Halin graphs constructed from the same tree using different embedding.}
        \label{fig:non-unique-Halin}
\end{figure}

\noindent
Based on the structure of Halin graphs, one may suspect that they inherit many of properties of their underlying tree.
Accordingly in the following subsection we present some properties and well-known facts regarding the OLA of trees.
\subsection{Linear Arrangement of Trees}
\label{sec:LA-trees}
As one of non-trivial results, the OLA of trees was first shown to be polynomially solvable in~\cite{goldberg1976algorithm}
and more efficient algorithms were later presented in~\cite{shiloach1979, chung1984}.
In this section we present some well-known properties of OLA of trees which
simplify the understanding of linear arrangement of Halin graphs
in the rest of the report. For more details, one can refer to ~\cite{chung1984}.
\begin{property}
\label{prop:tree-end-nodes}
Given an OLA $\vfi^ *$ for tree $T=(V,E)$, vertices that are assigned label $1$ and $n$ are leaf nodes (the two extreme vertices are leafs in $T$).
\end{property}
\begin{definition}{Spinal path}
\label{def:tree-spine}
Given layout $\vfi$ for a tree $T=(V,E)$, and vertices $v,u \in V$
where $\vfi(v) = 1$ and $\vfi(u) = n$, we define the path $P = (w_1 = v, w_2, \dots, w_l = u)$ connecting
$v$ and $u$, the \textbf{spine} of tree $T$ corresponding to $\vfi$.
\end{definition}

\begin{property}
\label{prop:tree-monotone-path}
Having OLA $\vfi^ *$ for tree $T=(V,E)$ and spinal path $P = (w_1 = v, w_2, \dots, w_l = u)$, it is the case that $\forall 1 \le i < l: \vfi^ *(w_i) < \vfi^ *(w_{i+1})$. In other word, the function $\vfi^ *$ is monotonic along the path $P$.
\end{property}
\begin{definition}{Spinal rooted subtree and anchored branches}
\label{def:anchored-subtree}
Given layout $\vfi$ for a tree $T=(V,E)$ and spinal path $P = (w_1 = v, w_2, \dots, w_l = u)$,
Removing all edges of $P$, leaves us with a set of subtrees $T_1, T_2, \ldots, T_l$ respectively rooted at
$w_1, w_2, \dots, w_l$. Removing spinal vertex $w_i$ from $T_i$ where $d_T(w_i) = k > 2$, results in a set of branches $B_{i,1}, B_{i,2},\ldots,B_{i,k-2}$. Each branch $B_{i,j}$ is anchored at a vertex $v_j$ which is connected to $w_i$.
\end{definition}
\begin{property}
\label{prop:tree-subtrees-labeled-separately}
Consider a tree $T=(V,E)$ and its OLA $\vfi^ *$ and the corresponding spinal path $P = (w_1 = v, w_2, \dots, w_l = u)$.
Removing all the edges of $P$ results in a set of $l$ subtrees $T_1, \ldots, T_l$, respectively rooted at $w_1, \ldots, w_l$.
Then based on $\vfi^ *$, for a fixed $i$, the vertices of of $T_i$ are labeled by contentious integers. Formally speaking:
\[
    \forall 1 < i < l, \forall u \in T_{i-1}, v \in T_{i}, w \in T_{i+1} \Rightarrow \vfi^ *(u) < \vfi^ *(u) < \vfi^ *(w)
\]
Moreover $\vfi^ *$ restricted to $V(T_i)$ (denoted by $\vfi^ *_i$) is optimal for $T_i$.
\end{property}
%
%
\begin{comment}
\begin{property}
\label{prop:tree-branches-labeled-separately}
Consider a tree $T=(V,E)$ and the corresponding OLA $\vfi^ *$ and subtrees $T_1, \ldots, T_l$ after removing the edges of spinal path $P = (w_1 = v, w_2, \dots, w_l = u)$. Let $\Set{B_{i,1},\dots,B_{i,k-2}}$ be the set of \emph{branches} connected to a spinal vertex $w_i$
with degree $k > 2$. The vertices of each branch $B_{i,j}$ are labeled with continuous integers.
\end{property}
\end{comment}
%
%

\section{Some Properties of OLA of Halin Graphs}
\label{sect:OLA-halin-graphs}
    Halin graphs are the example of edge-minimal 3-connected graphs.%~\cite{halin connectiviti}.
Hence, in a Halin graph $H = T \uplus C$, for any two vertices $v$ and $u$, there are exactly three edge-disjoint paths connecting $v$ and $u$
where one comprises only edges of $E(T)$.
\begin{definition}{Spinal path in Halin graphs}
Given layout the $\vfi$ for a Halin graph $H$ and two vertices $v,u \in E(H)$ where $\vfi(v)=1$ and $\vfi(u)=n$,
the \emph{spinal path} based on $\vfi$, is defined as the path $P = (w_1 = v, w_2, \dots, w_l = u)$ where for every $1 \le i < l$,
$\Set{w_i,w_{i+1}}$ is an edge in $T$.
\end{definition}
\begin{definition}{Spinal rooted subtree and anchored branches in Halin graphs}
\label{def:anchored-subtree}
Given layout $\vfi$ for a a Halin graph $H=T \uplus C$ and spinal path $P = (w_1 = v, w_2, \dots, w_l = u)$,
removing all edges of $P$ and $E(C)$ results in a set of subtrees $T_1, T_2, \ldots, T_l$, respectively rooted at
$w_1, w_2, \dots, w_l$. Removing spinal vertex $w_i$ from $T_i$ where $d_T(w_i) = k > 2$, give us a set of branches $B_{i,1}, B_{i,2},\ldots,B_{i,k-2}$. Also each branch $B_{i,j}$ is anchored at a vertex $v_j$, connected to $w_i$.
\end{definition}
\begin{lemma}
\label{lemma:halin-subtrees-labeled-separately}
Consider the Halin graph $H=T \uplus C$ and the spinal path $P = (w_1, w_2, \dots, w_l)$ based on a given OLA $\vfi^ *$.
Removing all the edges of $P$ and $C$ results in a set of $l$ subtrees $T_1, \ldots, T_l$ respectively rooted at $w_1, \ldots, w_l$.
For a fixed $i$, the vertices of of $T_i$ are labeled by contentious integers by OLA $\vfi^ *$. Formally speaking:
\[
    \forall 1 < i < l, \forall u \in T_{i-1}, v \in T_{i}, w \in T_{i+1} \Rightarrow \vfi^ *(u) < \vfi^ *(u) < \vfi^ *(w)
\]
\end{lemma}
See proof~\ref{proof:halin-subtrees-labeled-separately} in Appendix~\ref{apndx:Proof-lemmas}.
\begin{corollary}
\label{corollary:halin-monotonic-on-spine}
Having OLA $\vfi^ *$ for Halin graph $H$ and spinal path $P = (w_1, w_2, \dots, w_l)$, it is the case that $\forall 1 \le i < l: \vfi^ *(w_i) < \vfi^ *(w_{i+1})$. In other word, the function $\vfi^ *$ is monotonic along the path $P$.
\end{corollary}
\begin{lemma}
\label{lemma:halin-branches-labeled-separately}
Consider an OLA $\vfi^{*}$ for a Halin graph $H = T \uplus C$ and the set of subtrees $T_1, \ldots, T_l$ resulted after removing the edges of $C$ and the spinal path $P = (w_1, w_2, \dots, w_l)$. Let $\Set{B_{i,1},\dots,B_{i,k-2}}$ be the set of branches of $T_{i}$ connected to a spinal vertex $w_i$ with degree $k > 2$. For two branches $B_{i,j}$ and $B_{i,j'}$:
\begin{itemize}
  \item If $\vfi^{*}$ is of type $(\ldots , w_i , \ldots , V(B_{i,j}) \cup V(B_{i,j'}) , \ldots)$ then it is of type $(\ldots , w_i , \ldots , V(B_{i,j}), V(B_{i,j'}) , \ldots)$ or $(\ldots , w_i , \ldots , V(B_{i,j'}), V(B_{i,j}) , \ldots)$
  \item If $\vfi^{*}$ is of type $(\ldots, V(B_{i,j}) \cup V(B_{i,j'}), \ldots, w_i, \ldots)$ then it is of type $(\ldots, V(B_{i,j}), V(B_{i,j'}), \ldots, w_i , \ldots)$ or $(\ldots, V(B_{i,j'}), V(B_{i,j}), \ldots, w_i, \ldots)$
\end{itemize}
\noindent
In other word the two branches $B_{i,j}$ and $B_{i,j'}$ which are on the same side of $w_i$ (either their vertices are all labeled after $w_i$ or all before it), do not overlap.
\end{lemma}
For proof refer to Appendix~\ref{apndx:Proof-lemmas}, proof~\ref{proof:halin-branches-labeled-separately}.
\begin{theorem}
\label{thrm:halin-OLA-extreme-vetrices}
Given an OLA $\vfi^{*}$ for a Halin graph $H=T \uplus C$ and the vertices $v$ and $u$ where $\vfi^{*}(v) = 1$ and $\vfi^{*}(u) = n$, it is always the case that:
\begin{itemize}
  \item $v$ and $u$ are both leaves in $T$ or
  \item if $v$ (or $u$) is not a leaf vertex in $T$, then degree of $v$ is three and it is connected to
  exactly two leaves in $T$. Accordingly replacing the label of $v$ (or $u$) with one of it's leaf nodes, we get
  another OLA $\vfi^{\circledast}$ where the extreme nodes are leaves in $T$.
\end{itemize}
\end{theorem}
This lemma is proven in proof~\ref{proof:halin-OLA-extreme-vetrices} in Appendix~\ref{apndx:Proof-lemmas}.
\begin{corollary}
\label{corol:halin-opt-LA-bound}
Consider an OLA $\vfi^{*}$ for a Halin graph $H=(V,E)$ constructed from tree $T = (V,E')$ and cycle $C$, then:
\[
    LA(\vfi^{*},H) \ge 2 \times (n-1) + LA(\vfi_T ^{*},T)
\]
where $\vfi_T ^{*}$ is the OLA for $T$.
\end{corollary}

\section{Halin Graphs With Polynomially Solvable LA Algorithm}
\label{sect:halin-solvable-case}
    As mentioned before, the OLA problem is polynomially solvable for trees.
The OLA of a Halin graph $H=T \uplus C$, depends both on the underlying tree and
the planar embedding of $T$.
Motivated by the work in~\cite{shiloach1979}, in this section we study some classes of Halin graphs where OLA problem
can be solved in polynomial time. More specifically we show that for these classes
of Halin graphs, the equality in corollary~\ref{corol:halin-opt-LA-bound} holds.
\begin{definition}{Recursively Balanced Trees}
\label{def:RBT}
Consider the tree $T$ and the vertex $v_r$, designated as the root of the tree,
and the set of vertices $v_{r,0}, \ldots, v_{r,k}$ connected to the $v_r$ as it's direct children.
Removing the set of edges $\Set{v_r,v_{r,0}}, \ldots, \Set{v_r,v_{r,k}}$ results in the set of subtrees $T_{r,0}, \ldots, T_{r,k}$,
respectively rooted at $v_{r,0}, \ldots, v_{r,k}$.
$T$ is recursively balanced if:
\begin{itemize}
  \item $\mathcal{T}_{r,0} = \mathcal{T}_{r,1}= \ldots = \mathcal{T}_{r,k}$
  \footnote{$\mathcal{T}_i = |V(T_i)|$. See notation~\ref{not:subtee-size} in Appendix~\ref{apndx:Proof-lemmas}.} and
  \item $T_{r,i}$, rooted at $v_{r,i}$, is recursively balanced for $i = 0,1, \ldots,k$
\end{itemize}

\end{definition}
\noindent
The root vertex $v_r$ of a Recursively Balanced Tree (RBT), is the only vertex satisfying the properties of the \emph{central vertex} in the following theorem.
\begin{theorem}
\label{thr:centeral-vertex}
Given a tree $T = (V,E)$, there exists a vertex $v_r$  where the set of subtrees $T_0,\ldots, T_k$ resulted by removing $v_r$ from $T$,
satisfies:
\begin{align*}
    \mathcal{T}_i \le \lfloor \dfrac{n}{2} \rfloor & \quad \text{for} \quad i = 0, \ldots, k
\end{align*}
\end{theorem}
\noindent
For proof see~\cite{shiloach1979}.

Considered a tree $T$ rooted at $v_r$ and the corresponding subtrees $T_0,\ldots, T_k$ after removing $v_r$,
where $\mathcal{T}_0\ge \mathcal{T}_2 \ge \ldots \ge \mathcal{T}_k$.
Assume that an OLA $\vfi^{*}$ for $T$ is of type $(\ldots,T_i,\ldots, v_r, \ldots)$ or $(\ldots,v_r,\ldots, T_i, \ldots)$ for some subtree $T_i$.
$T_i$ is called (respectively right or left) \emph{anchored} subtree, rooted at $v_i$ connecting $T_i$ to $v_r$.
A tree $T$ which is not anchored is called a \emph{free} tree.
In theorem~\ref{thr:motiv-shiloach-thoerm}, which is the motivating theorem and the heart of the OLA algorithm for trees in~\cite{shiloach1979}, we show the root of a tree by $v_r$. Vertex $v_r$ is the central vertex in the case of free trees, or the anchor vertex if the tree is an anchored subtree. Also the parameter $\alpha$ is $0$ for free trees and $1$ otherwise.
\begin{theorem}
\label{thr:motiv-shiloach-thoerm}
Given a free or (right) anchored tree $T = (V,E)$
\footnote{The theorem symmetrically holds in case of left anchored subtrees.},
let $\rho$ be the largest integer that satisfies the following:
\[
\mathcal{T}_i > \lfloor \dfrac{\mathcal{T}_1+2}{2} \rfloor + \lfloor \dfrac{\mathcal{T}_*+2}{2} \rfloor
\quad  \text{for} \quad  i=1, \ldots, 2\rho-\alpha
\]
 where:
\[
    \mathcal{T}_* = n - \sum\limits_{0}^{2\rho-\alpha} \mathcal{T}_i
\]
\begin{itemize}
  \item If $\rho = 0$, the OLA of $T$ is of type $(T_0,v_r,\ldots)$
  \item If $\rho > 0$, then $T$ has an OLA of type either $(T_0,v_r,\ldots)$ or $(T_1,T_3,\ldots, T_{2\rho-1},\ldots, v_r,\ldots, T_{2\rho-2\alpha},\ldots,T_4,T_2)$
\end{itemize}
\end{theorem}
\begin{notxxx}
Consider the layout $\vfi$ of type $(T_1, \ldots, T_i,\ldots,T_j, \ldots, T_k)$ for the tree $T$ rooted at the central vertex $v_r$. Swapping the arrangement of vertices of two subtrees $T_i$ and $T_j$, while keeping the relative order of the vertices of each subtree unchanged (or reversed), is presented using
operator $\sigma(\vfi, T_i,T_j)$ which is of type $(T_1, \ldots, T_j,\ldots,T_i, \ldots, T_k)$.
%\noindent
%$\vfi^{\prime} = \sigma(\vfi, T_i,T_j)$ is of type $(T_1, \ldots, T_j,\ldots,T_i, \ldots, T_k)$ if based on $\vfi$, both $T_i$ and $T_j$ are arranged on the same side of $v_r$, and of type $(T_1, \ldots, \overleftarrow{T_j},\ldots,\overleftarrow{T_i}, \ldots, T_k)$ otherwise. $\overleftarrow{T_j}$ (and $\overleftarrow{T_i}$) denotes that the order of vertices in $T_i$ (and $T_j$) is reversed.
\end{notxxx}
\begin{lemma}
\label{lemma:LA-of-rec-balanced-trees}
Given a recursively balanced tree $T$ rooted at the central vertex $v_r$, and the corresponding subtrees $T_{r,1}, \ldots, T_{r,k}$, there exists
an OLA $\vfi^{*}$ of type $(T_{r,1}, \ldots, T_{r,\overline{k}},v_r,T_{r,\overline{k}+1}, \ldots, T_k)$,
where $\overline{k} = \lceil \dfrac{k+1}{2} \rceil$.
\end{lemma}
\begin{proof}
We know that the subtrees $T_{r,1}, \ldots, T_{r,k}$ have the same size and $v_r$ satisfy the central vertex theorem~\ref{thr:centeral-vertex},
Accordingly, considering the theorem~\ref{thr:motiv-shiloach-thoerm},
it's easy to see that there exists an OLA $\vfi_0$ where half of subtrees are labeled before $v_r$ and the other half are
labeled after $v_r$ while the vertices of no two subtrees overlap.

\noindent
Based on the structure of $\vfi_0$,
there exists a sequence of layouts $(\vfi_0,\vfi_1,\ldots,\vfi_l=\vfi^{*})$ where for $k = 1, \ldots, l$,  $\vfi_k = \sigma(\vfi_{k-1}, T_{r,i},T_{r,j})$
for some subtrees $T_{r,i}$, $T_{r,j}$. Since all the subtrees have the same size, then $LA(\vfi_0,T) = LA(\vfi_1,T)= \ldots = LA(\vfi^{*},T)$.

\noindent
Generally, given an OLA $\vfi^{*}$ of type $(T_{r,1}, \ldots, T_{r,\overline{k}},v_r,T_{r,\overline{k}+1}, \ldots, T_k)$ for the RBT $T$ and two subtrees $T_{r,i}$ and $T_{r,j}$, $\sigma(\vfi^{*}, T_{r,i},T_{r,j})$ is also an OLA for $T$.
\end{proof}

\begin{lemma}
\label{lemma:spinal-path-of-rec-balanced-trees}
Let $T_{r,1}, \ldots, T_{r,k}$ be the set of subtrees of the RBT $T$ resulted by removing the root vertex $v_r$.
Given two leaf vertices $v \in V(T_{r,i}), u \in V(T_{r,j})$ for $i \neq j$, the simple path $P=(v,\ldots,v_r,\ldots,u)$
connecting $v$ and $u$ (via $v_r$) is the spinal path for some OLA $\vfi^{*}$. In other word there is an OLA $\vfi^{*}$, where $\vfi^{*}(v) = 1$ and $\vfi^{*}(u) = n$.
\end{lemma}
\begin{proof}
An immediate result of lemma~\ref{lemma:LA-of-rec-balanced-trees} is that there exists an OLA of type $(T_{r,i}, \ldots,v_r, \ldots, T_j)$ for tree $T$.
Also note that the two subtrees $T_{r,i}$ and $T_{r,j}$ are recursively balanced trees. Applying the same approach recursively and excluding all the the details, one can deduce that there exists an OLA for $T$ which is of type $(v, \ldots,v_r, \ldots, u)$.
\end{proof}
\begin{exaxxx}
Figure~\ref{fig:RBT-exp-and-OLA} depicts an example of a recursively balance tree (in~\ref{fig:RBT-exp}) and it's corresponding OLA $\vfi^{*}$(in~\ref{fig:RBT-exp-OLA}).
As you see the operation $\sigma(\vfi^{*}, T_1,T_2)$ will result in another layout with the same value.
Generally, given an OLA $\vfi^{*}$ for a recursively balanced tree $T = (V,E)$
and $v \in V$ and any two rooted subtrees $T_{i,j}$ and $T_{i,j'}$ connected to $v$, it is the case that $\sigma(\vfi^{*}, T_{i,j},T_{i,j'})$ is also an OLA.
\begin{figure}
\centering
        \begin{subfigure}[b]{0.4\textwidth}
                \includegraphics[scale=.5]
                {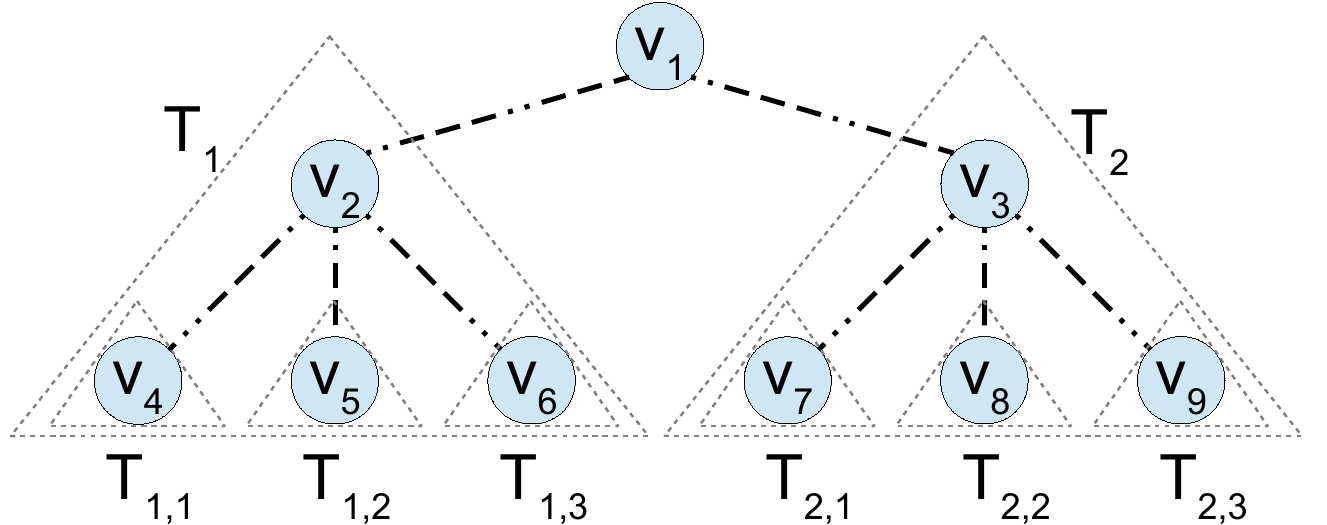}
                \caption{An example of a recursively balance tree $T$ rooted at $v_1$. The two subtrees $T_1$ and $T_2$ of $v_1$ are highlighted by larger enclosing triangles.}
                \label{fig:RBT-exp}
        \end{subfigure}
        \qquad\qquad\qquad
        \begin{subfigure}[b]{0.4\textwidth}
                \includegraphics[scale=.5]
                {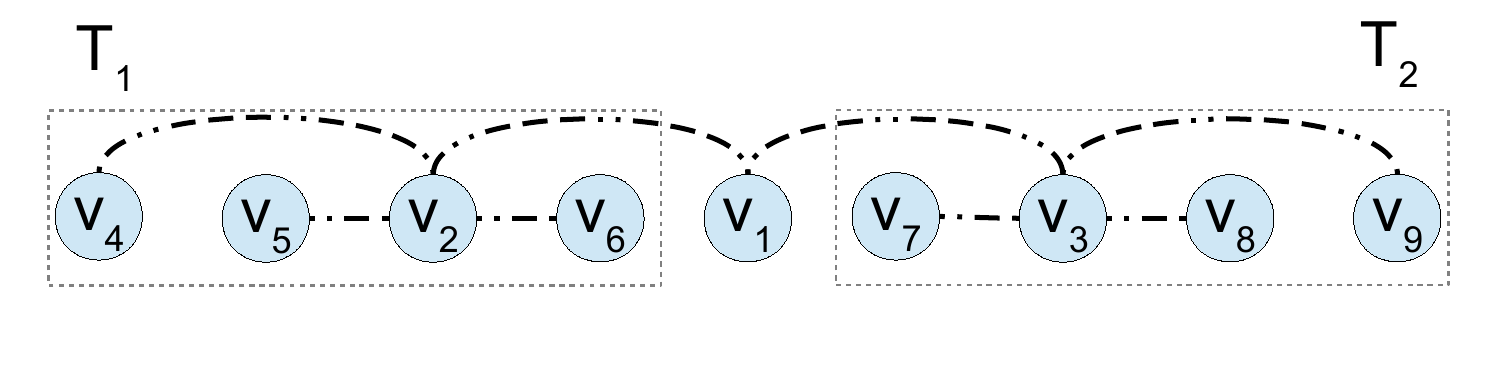}
                \caption{An OLA $\vfi^{*}$ for tree $T$.}
                \label{fig:RBT-exp-OLA}
        \end{subfigure}
        \caption{An example of a recursively balance tree and it's corresponding OLA.}
        \label{fig:RBT-exp-and-OLA}
\end{figure}
\end{exaxxx}
Following the results of lemmas~\ref{lemma:LA-of-rec-balanced-trees} and~\ref{lemma:spinal-path-of-rec-balanced-trees},
In following we present an approach to find an OLA for recursively balanced tree in linear time.
\begin{theorem}
Having a recursively balanced tree $T=(V,E)$ rooted at $v_r$, an OLA $\vfi^{*}$ for $T$ can be found in linear time.
\end{theorem}
\begin{proof}
 Assume $T_{r,1}, \ldots, T_{r,k}$ are the subtrees connected to $v_r$ respectively via $v_{r,1},v_{r,2},\ldots,v_{r,k} \in V$.
 From lemma~\ref{lemma:LA-of-rec-balanced-trees} an OLA $\vfi^{*}$ of type
 $(T_{r,1}, \ldots, T_{r,\overline{k}},v_r,T_{r,\overline{k}+1}, \ldots, T_k)$ exists for $T$.
 Each subtree has exactly $\frac{|V|-1}{k}$ vertices, hence $\vfi^{*}(v_r) = \overline{k}\times \frac{|V|-1}{k} + 1$, where $\overline{k} = \lceil \frac{k+1}{2} \rceil$.

\noindent
Also based on the definition~\ref{def:RBT} every subtree $T_{r,i}$, for $1 \le i \le k$, is recursively balance with the central vertex $v_{r,i}$.
Therefore, using the same approach one can go on with constructing OLA $\vfi^{*}$ by finding the label of $v_{r,i}$, for $1 \le i \le k$.
Applying this method recursively OLA $\vfi^*$ can be found while every vertex of the tree is visited $O(1)$ times.
\end{proof}
The following two theorems conclude this section by presenting some classes of Halin graphs which there exists a polynomial OLA
algorithm for them. More specifically, given a Halin graph $H = T \uplus C$ from these classes,
an OLA for $H$ can be derived given any optimal layout for the underlying tree $T$ \footnote{Remember that OLA problem is polynomially solvable for trees.}.
\begin{theorem}
Consider a Halin graph $H = T \uplus C$, where the underlying tree $T$ is recursively balanced, rooted at $v_r$.
Let $\vfi^{\circledast}$ be an OLA for $T$.
There exist a linear arrangement $\vfi^{*}$ s.t.
\begin{itemize}
  \item $LA(\vfi^{*}, H) = LA(\vfi^{\circledast}, T) + 2\times (|V|-1)$. Hence, based on corollary~\ref{corol:halin-opt-LA-bound}, $\vfi^{*}$ is an OLA for $H$
  \item $\vfi^{*}$ can be constructed from $\vfi^{\circledast}$ in $O(|V| \log{|V|})$
\end{itemize}
\end{theorem}
\begin{proof}
We know that for every layout $\vfi$, it is the case that $LA(\vfi, C) \ge 2\times (n-1)$, where $n = |V|$. Hence, given the OLA $\vfi^{\circledast}$ for $T$, if $LA(\vfi^{\circledast}, H) = LA(\vfi^{\circledast}, T) + 2\times (n-1)$, then $\vfi^{\circledast}$ is also an
OLA for $H$ as well.
Otherwise, starting from $\vfi^{\circledast}$, we present an iterative approach where using a sequence of swapping operations, an OLA is found for $H$. In this sequence of swapping, after each swap operation the value of arrangement stays unchanged for $T$, and decreases for $H$.

\noindent
This procedure is presented in algorithm~\ref{alg:redundant-crossing-elim}.
Assume the underlying tree $T$, rooted at central vertex $v_r$, has height $\hbar$ \footnote{We consider the height of a tree with consist of only vertex is $1$.}
\begin{algorithm}
\caption{Finding OLA $\vfi^{*}$ for Halin graph $H=(T,C)$ given OLA $\vfi^{\circledast}$ for RBT $T$}
\label{alg:redundant-crossing-elim}
\begin{algorithmic}[1]
\State $\vfi^{*} \leftarrow \vfi^{\circledast}$
\State let $\Set{T_1,\ldots,T_k}$ be subtree of height $\hbar-1$ as $\vfi^{*}$ is of type $(T_1,\ldots,v_r,\ldots,T_k)$
\State let $T_L$ be $T_1$ \label{algLine:TL}
\State let $T_R$ be one of the two subtrees connected to $T_L$ via $E(C)$ \label{algLine:TR}
\State $\sigma(\vfi^{*}, T_R,T_k)$
\State \textbf{for} $i = 1$ to $k-2$ \label{algLine:For1}
\State \quad Let $T_{i,R} \in \Set{T_{i+1},\ldots,T_{k-1}}$ be the subtree connected to $T_i$ via and edge in $E(C)$
\State \quad $\sigma(\vfi^{*}, T_{i+1},T_{i,R})$ \label{algLine:Part1}
\item[]
\State \textbf{for} $h = \hbar$ to $2$: \label{algLine:Part2-start}
\State \quad\textbf{for} every subtree $T$ of height $h$ rooted at $v_r$:
\State \quad\quad Let $(T_{r,1},\ldots,v_r,\ldots,T_{r,k})$ be $\vfi^*$ restricted to $T$
\State \quad\quad\textbf{if} $\vfi^{*}$ is of type $(T,\ldots)$:
\State \quad\quad\quad \textbf{ReArrLeftSubTree(}$(T_{r,1},\ldots,T_{r,k})$, $\vfi^*$\textbf{)}
\State \quad\quad\textbf{else if} $\vfi^{*}$ is of type $(\ldots,T)$:
\State \quad\quad\quad \textbf{ReArrRightSubTree(}$(T_{r,1},\ldots,T_{r,k})$, $\vfi^*$\textbf{)}
\State \quad\quad\textbf{else}:
\State \quad\quad\quad \textbf{ReArrMidSubTree(}$(T_{r,1},\ldots,T_{r,k})$, $\vfi^*$\textbf{)}
\label{algLine:Part2-end}
\item[]
\State \textbf{ReArrMidSubTree(}$(T_{r,1},\ldots,T_{r,k})$, $\vfi^*$\textbf{)}:
\State \quad Let $T_{r,L}$ be the subtree connected to some vertex $v_L$, where based on $\vfi^{*}$, $v_L$ is labeled before $T_{r,1}$
\footnotemark{Hence, $v_L \notin \Set{V(T_{r,1}) \uplus \ldots,\uplus T_{r,k}}$ and $\forall v in \Set{V(T_{r,1}) \uplus \ldots,\uplus T_{r,k}}, \vfi^{*}(v_L) < \vfi^{*}(v)$}
\State \quad Let $T_{r,R}$ be the subtree connected to some vertex $v_R$, where based on $\vfi^{*}$, $v_R$ is labeled after $T_{r,k}$
\footnotemark{Hence, $v_R \notin \Set{V(T_{r,1}) \uplus \ldots,\uplus T_{r,k}}$ and $\forall v in \Set{V(T_{r,1}) \uplus \ldots,\uplus T_{r,k}}, \vfi^{*}(v) < \vfi^{*}(v_R)$}
\label{algLine:TrR}
\State \quad $\sigma(\vfi^{*}, T_{r,1},T_{r,L})$
\State \quad $\sigma(\vfi^{*}, T_{r,k},T_{r,R})$
\State \quad \textbf{if} $k > 3$:
\State \quad\quad \textbf{ReArrMidSubTree(}$(T_{r,2},\ldots,T_{r,k-1})$, $\vfi^*$\textbf{)}
\item[]
\State \textbf{ReArrLeftSubTree(}$(T_{r,1},\ldots,T_{r,k})$, $\vfi^*$\textbf{)}:
\State \quad Let $T_{r,L}$ be the subtree connected to $T_R$ via $E(C)$
\State \quad $\sigma(\vfi^{*}, T_{r,1},T_{r,L})$
\State \quad\textbf{for} $i = 1$ to $k-2$
\State \quad\quad Let $T_{i,R} \in \Set{T_{i+1},\ldots,T_k}$ be the subtree connected to $T_i$
\State \quad\quad $\sigma(\vfi^{*}, T_{i+1},T_{i,R})$
\item[]
\State \textbf{ReArrRightSubTree(}$(T_{r,1},\ldots,T_{r,k})$, $\vfi^*$\textbf{)}:
\State \quad Let $T_{r,R}$ be the subtree connected to $T_L$ via $E(C)$
\State \quad $\sigma(\vfi^{*}, T_{r,k},T_{r,R})$
\State \quad\textbf{for} $i = k$ to $2$
\State \quad\quad Let $T_{i,R} \in \Set{T_1,\ldots,T_{i-1}}$ be the subtree connected to $T_i$
\State \quad\quad $\sigma(\vfi^{*}, T_{i-1},T_{i,R})$
\end{algorithmic}
\end{algorithm}
\paragraph{Correctness of the algorithm.}
Starting with $\vfi^{*} =  \vfi^{\circledast}$, after execution of
line~\ref{algLine:Part1}, we will end up with a potentially modified layout $\vfi^{*}$ of type
$(T_1,\ldots,v_r,\ldots,T_k)$ where $T_1$ is directly connected to $T_k$ via an edge from $E(C)$
and for $1 \le i < k$, $T_i$ is directly connected to $T_{i+1}$ through $E(C)$.
So if we collapse every subtree $T_i$ for $1 \le i < k$ into one vertex, the resulted $\vfi^{*}$
is an OLA for the corresponding Halin graph.

\noindent
So far, Based on the resulted layout $\vfi^{*}$, $T_L = T_1$ and $T_R = T_k$, respectively defined in lines~\ref{algLine:TL} and~\ref{algLine:TR},
are the two left and right boundary subtrees and all other subtrees are middle subtrees.

\noindent
Lines~\ref{algLine:Part2-start} to~\ref{algLine:Part2-end} of algorithm,
guarantee that in a recursive approach, for every subtree $T$ of height $1 \le h < \hbar$, based on the final $\vfi^{*}$ :
\begin{itemize}
  \item If $T$ is a left side subtree (\emph{i.e.} if $\vfi^{*}$ is of type $(T,\ldots)$), consider $v \in V(T)$, where $\vfi^{*}(v) = 1$. $v$ is connected
  to $T_R$ via $e \in E(C)$
  \item If $T$ is a right side subtree (\emph{i.e.} if $\vfi^{*}$ is of type $(\ldots,T)$), consider $v \in V(T)$, where $\vfi^{*}(v) = |V|$. $v$ is connected
  to $T_L$ via $e \in E(C)$
  \item Otherwise $\vfi^{*}$ is of type $(\ldots,T_1,T,T_2,\ldots)$. Let $v_L \in V(T)$ be the vertex of $V$
  s.t. $\forall v \in V(T), \vfi^{*}(v_L) \le \vfi^{*}(v)$. Similarly let $v_R \in V(T)$ be the vertex
  s.t. $\forall v \in V(T), \vfi^{*}(v) \le \vfi^{*}(v_R)$. Then $v_L$ and $v_R$ are respectively directly connected to $T_1$ and $T_2$ through $E(C)$
\end{itemize}
Therefore it can be inferred that based on the final layout $\vfi^{*}$, $LA(\vfi^{*},H) = LA(\vfi^{*},T)  +  LA(\vfi^{*},C) = LA(\vfi^{*},T) + 2 \times (n-1)$. But we know that the swap operation $\sigma$ does not change
the value of linear arrangement for the underlying tree $T$.
Hence $LA(\vfi^{*},H) = LA(\vfi^{\circledast},T) + 2 \times (n-1)$ which induces the optimality of $\vfi^{*}$.
\paragraph{Time complexity analysis.} The time complexity of the algorithm depends on the two
major \textbf{For} loops in lines~\ref{algLine:For1} and~\ref{algLine:Part2-start}.
\begin{itemize}
  \item
    \textbf{Analysis of the first loop in line~\ref{algLine:For1}.} Assuming every basic swap operation
    is an atomic operation with cost $O(1)$, then the cost of every
    $\sigma(\vfi^{*}, T_1,T_2)$ is $O(|V(T_1)|)= O(|V(T_2)|)$. Hence, the cost of the loop at line~\ref{algLine:For1} is
    $O(k \times \dfrac{n-1}{k}) = O(n)$ \footnote{Every subtree $T_i$ has exactly $\dfrac{n-1}{k}$ vertices}.

  \item
    \textbf{Analysis of the second loop in line~\ref{algLine:Part2-start}.}
    At every iteration, if there are $k$ subtrees $\Set{T_{r,1},\ldots,T_{r,k}}$,
    In worst-case scenario at most $O(k)$ swap operations $\sigma$ are carried out. For $1 \le i \le k$, $|V(T_{r,i})| = \dfrac{n}{k}$ .Therefore the cost of each iteration is $O(k\times\dfrac{n}{k}) = O(n)$.

    \noindent
    Having the fact that $\hbar = O(\log{n})$, we conclude the time complexity of the loop in line~\ref{algLine:Part2-start} as $O(n \log{n})$, which dominates the time complexity of the whole algorithm.
\end{itemize}
\end{proof}
\begin{example}
\label{exp:non-RBT-layout}
In figure~\ref{fig:RBT-exp-and-OLA}, two layouts are presented for the Halin graph $H_2 = T \uplus C$ (Figure~\ref{fig:non-RBT-exp-2} in section~\ref{sect:prelim}).
Layout $\vfi_1$ in Figure~\ref{fig:non-RBT-exp-OLA-for-tree}, is an OLA for the underlying tree of
graph $H_2$ while it is not an optimal layout for $H_2$ itself (the OLA for $H_2$ is shown in Figure~\ref{fig:non-RBT-exp-OLA-for-graph}).
Enumerating all the OLAs for the underlying tree of $H_2$, it can be verified that non
is an OLA for $H_2$.

\noindent
$P=(v_7,v_2,v_2,v_4,v_{12})$ is the spinal path corresponding to the layout $\vfi_1$.
After removing the edges of the spinal path $P$ and  cycle $C$,
each vertex $v_i$ of the path $P$ corresponds to a subtree $T_i$.
Notice that subtree $T_1$, rooted at $v_1$, is not
an RBT. Hence based on the order of arrangement of the three branches connected to $v_1$, we may get different values
for the linear arrangement, and an ordering of the branches with the optimal arrangement values for $T$, is not necessary optimal
after adding the edge of cycle $C$ and path $P$ back.

\noindent
As opposed to the OLA of the underlying tree of $H_2$, given an OLA $\vfi$ for an arbitrary RBT $T'$,
and an arbitrary subtree $T_i$ rooted at some spinal vertex $v_i$,
all the branches of $T_i$ connected to $v_i$ have the same number of vertices and are
also recursively balanced. Hence for every Halin graph $H' = T' \uplus C'$ based on $T'$,
the layout $\vfi$ can be modified by changing the order of the branches of the subtrees where the value
of linear arrangement for $T'$ stays unchanged (Let's call the modified layout $\vfi'$),
while the value of linear arrangement for the edges of cycle $C'$
(\emph{i.e.} $\sum_{\Set{v,u} \in E(C')} \lambda(\Set{v,u}, \vfi')$)
is equal to $2 \times (n-1)$ \footnote{So there is no redundant crossing exists based on $\vfi'$.}.
Consequently the value of OLA for $H'$ is equal to $LA(\vfi, T) + 2 \times (n-1)$.
\end{example}
\begin{figure}
\begin{center}
        \begin{subfigure}[b]{0.65\textwidth}
                \includegraphics[scale=.5]
                {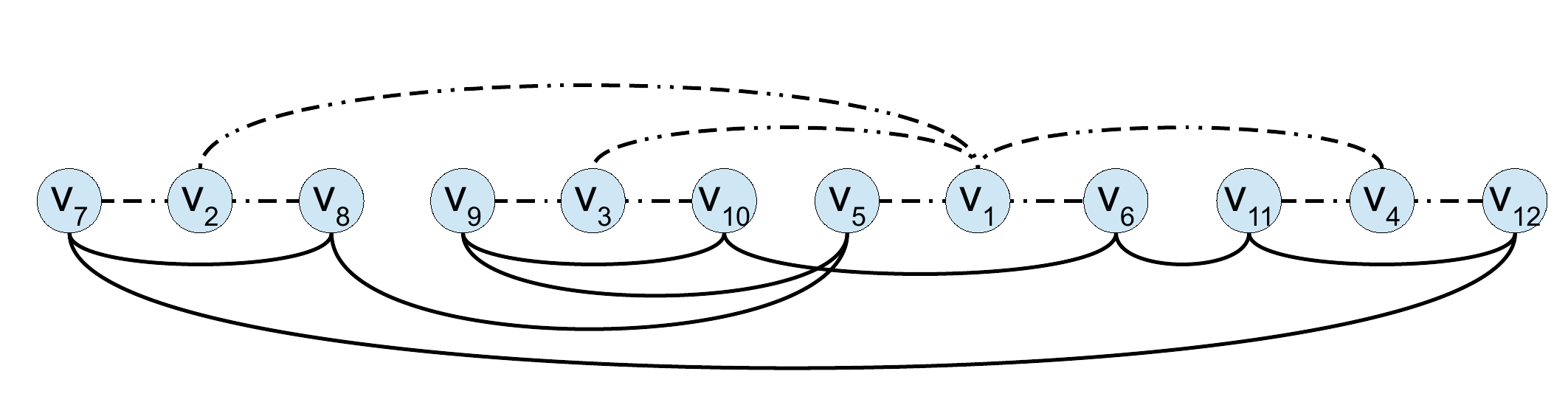}
                \caption{Layout $\vfi_1$ is an OLA for underlying tree of Halin graph $H_2$.}
                \label{fig:non-RBT-exp-OLA-for-tree}
        \end{subfigure}

        \begin{subfigure}[b]{0.65\textwidth}
                \includegraphics[scale=.5]
                {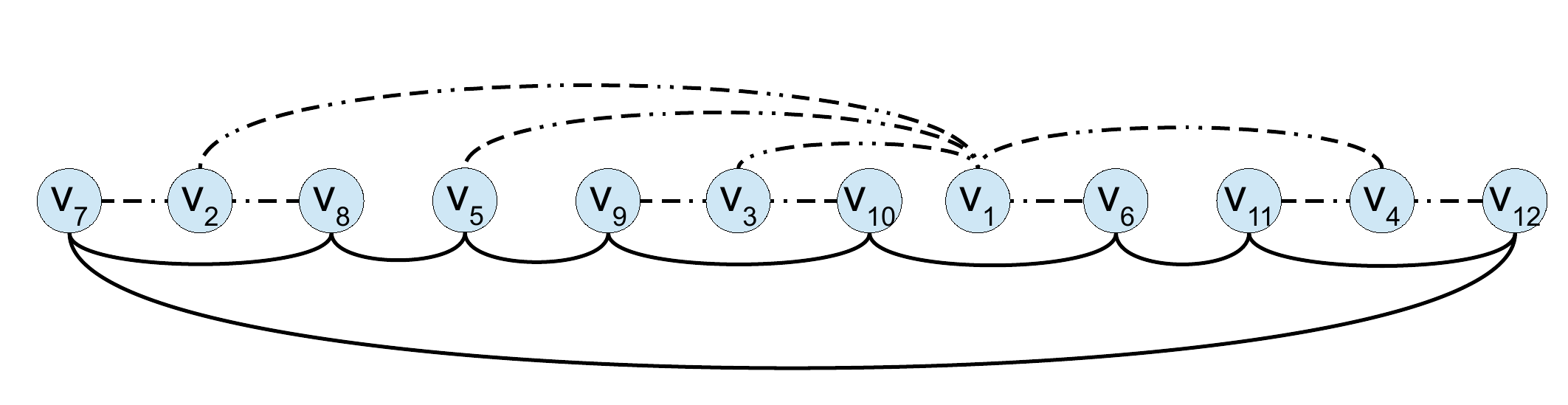}
                \caption{Layout $\vfi_2$ is an OLA for Halin graph $H_2$.}
                \label{fig:non-RBT-exp-OLA-for-graph}
        \end{subfigure}
        \caption{Two layout $\vfi_1$ and $\vfi_2$ for Halin graph $H_2 = T \uplus C$ (Figure~\ref{fig:non-RBT-exp-2}
        in section~\ref{sect:prelim}), where $LA(\vfi_1,H_2) > LA(\vfi_2,H_2)$ while $LA(\vfi_1,T) < LA(\vfi_2,T)$.
        More specifically $\vfi_1$ is an OLA for $T$ and $\vfi_2$ is the OLA for $H_2$.
        }
        \label{fig:non-RBT-exp-and-OLA}
\end{center}
\end{figure}
\begin{corollary}
\label{corol:halin-with-OLA-equal-to-tree}
Let $T$ be the underlying tree for some Halin graph $H = T \uplus C$ and let $\vfi^{\circledast}$ be an OLA for $T$
where $P=(w_1,\ldots,w_l)$ is respectively the spinal path and
based on $\vfi^{\circledast}$, and $\Set{T_1, \ldots, T_l}$ is  the set of subtrees remaining after removing all edges of $C$ and $P$.

\noindent
If for some OLA $\vfi^{\circledast}$ of $T$, $T_i$ rooted at $w_i$ is a recursively balance tree for $i= 1, \ldots, l$, then
there exists an OLA $\vfi^{*}$ for $H$ where $LA(\vfi^{*}, H) = LA(\vfi^{\circledast}, T) + 2\times (n-1)$.
\end{corollary}
Another class of Halin graphs which their underlying trees satisfy the sufficient property of corollary~\ref{corol:halin-with-OLA-equal-to-tree},
are the Halin graphs based on \emph{caterpillar trees}. This class of Halin graphs is studied in~\cite{easton1996solvable} and presented
result on value of their OLA, testifies the corollary~\ref{corol:halin-with-OLA-equal-to-tree}.

\section{Conclusion and Future Work}
\label{sect:future}
   As one of the simplest classes of non-outerplanar graphs,
in this work we studied some properties of OLA of Halin graphs
and we presented a lower bound for the value of OLA for Halin graphs.
We also introduced some classes of Halin graphs which the OLA can be
found in $O(n\log{n})$.
The problem of OLA of general Halin graphs is still open
and we believe a solution for the OLA of general Halin graphs
gives good insights into the properties of OLA of the
more general class of k-outerplanar graphs.

\Hide
{\footnotesize
\printbibliography
}

% \Hide
{\footnotesize % \small
\bibliographystyle{plainurl} % {plainnat} % {plain} % {alpha} % {siam} % {abbrv}
   % plain, abbrv, siam, alpha, are among dozens of other available styles
\bibliography{allbibs}
}
\newpage
\appendix
\section{Appendix: Proofs of Supporting Lemmas for Section~\ref{sect:OLA-halin-graphs}}
    \label{apndx:Proof-lemmas}
\begin{notxxx}
\label{not:subtee-size}
Consider the layout $\vfi$ for a Halin graph $H=T \uplus C$ and it's corresponding spinal path $P = (w_1, w_2, \dots, w_l)$.
$\mathcal{T}_i = |\Set{u \text{ s.t. } u \in V(T_i)}|$ presents the number of vertices of an spinal subtree $T_i$.
 Similarly for branch $B_{i,j}$, $\beta_{i,j} = |\Set{u \text{ s.t. }  u \in V(B_{i,j})}|$ stands for the number of vertices of branch $B_{i,j}$.
\end{notxxx}
\begin{notxxx}
\label{not:delta}
Given vertex $v \in V$, subset $\mathcal{V} \subseteq V$ and a layout $\vfi$, we define:
 \begin{align*}
 \delta_{\vfi} (\underline{v}, \overline{v}, \mathcal{V}) = &|\Set{u \text{ s.t. } u \in \mathcal{V} \text{ and } \vfi(\underline{v}) < \vfi(u) < \vfi(\overline{u})}|
 \end{align*}
 In other word, $\delta_{\vfi} (\underline{v}, \overline{v}, \mathcal{V})$ is the number of vertices in $\mathcal{V}$, which based on $\vfi$ are labeled with integers greater than the label of $\underline{v}$ and smaller than the label of $\overline{v}$.
Respectively $\delta_{\vfi} (-, \overline{v}, \mathcal{V})$ and $\delta_{\vfi} (\underline{v}, -, \mathcal{V})$ can be interpreted as the number of vertices of $\mathcal{V}$ labeled before $\overline{v}$ and after $\underline{v}$.
\end{notxxx}
In what follows we present an auxiliary lemma and its proof that will be helpful 
in simplifying and understanding of the proof of lemma~\ref{lemma:halin-subtrees-labeled-separately}.
\begin{lemma}
\label{lemma:halin-first-subtree-labeled-separately}
Consider the OLA $\vfi^ *$ for the Halin graph $H=T \uplus C$ and the corresponding spinal path $P = (w_1, w_2, \dots, w_l)$.
Removing all the edges of $P$ and $C$ results in a set of $l$ subtrees $T_1, \ldots, T_l$ respectively rooted at $w_1, \ldots, w_l$.
In the layout $\vfi^ *$, the vertices of $T_1$ are labeled by contentious integers and before all the vertices of $H - T_1$. Formally speaking:
\[
    \forall  v \in T_1, u \notin T_1 \Rightarrow \vfi^{*}(v) < \vfi^{*}(u)
\]
\end{lemma}
\begin{proof}
We prove this lemma by showing that the opposing assumption contradicts the optimality of $\vfi^{*}$. In other other word, if $\exists v \in V(T_1), u \notin V(T_1), \vfi^{*}(u) < \vfi^{*}(v)$, we suggest an alternative layout $\vfi^{\diamond}$ where $LA(\vfi^{*},H) > LA(\vfi^{\diamond}, H)$. Two layouts $\vfi^{*}$ and $\vfi^{\diamond}$ are respectively shown in Figure~\ref{fig:first-subtree-labeled-separately} and~\ref{fig:first-subtree-labeled-separately-after}.
\begin{figure}
        \centering
        \begin{subfigure}[b]{0.3\textwidth}
                \includegraphics[scale=.7]
                {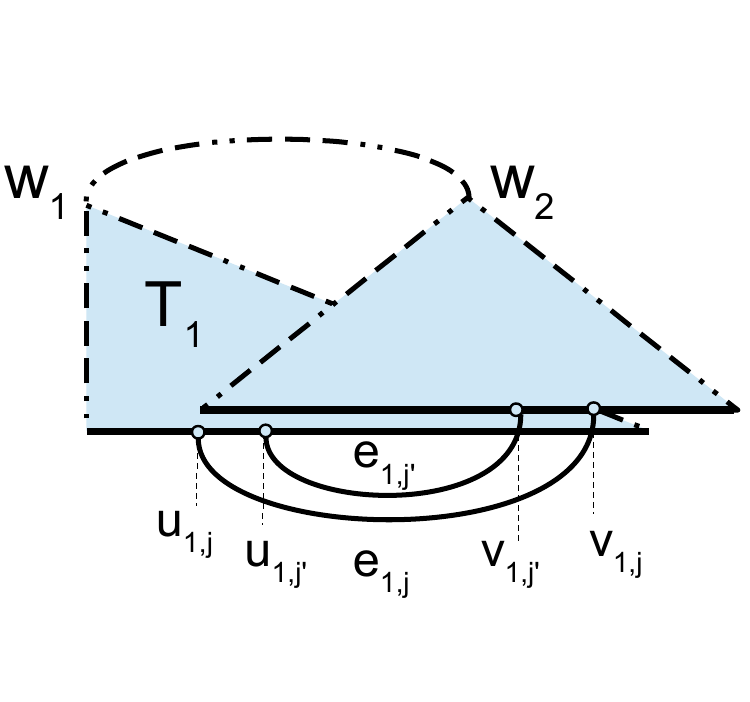}
                \caption{OLA layout $\vfi^{*}$ which labels some vertex $u \in V(T_1)$ with integers larger than the label of some vertex $v \in V(H) - V(T_1)$. In other word the arrangement of $T_1$ has overlapping area with arrangement of $H - T_1$.}
                \label{fig:first-subtree-labeled-separately}
        \end{subfigure}
        \qquad\qquad\qquad
        \begin{subfigure}[b]{0.3\textwidth}
                \includegraphics[scale=.7]
                {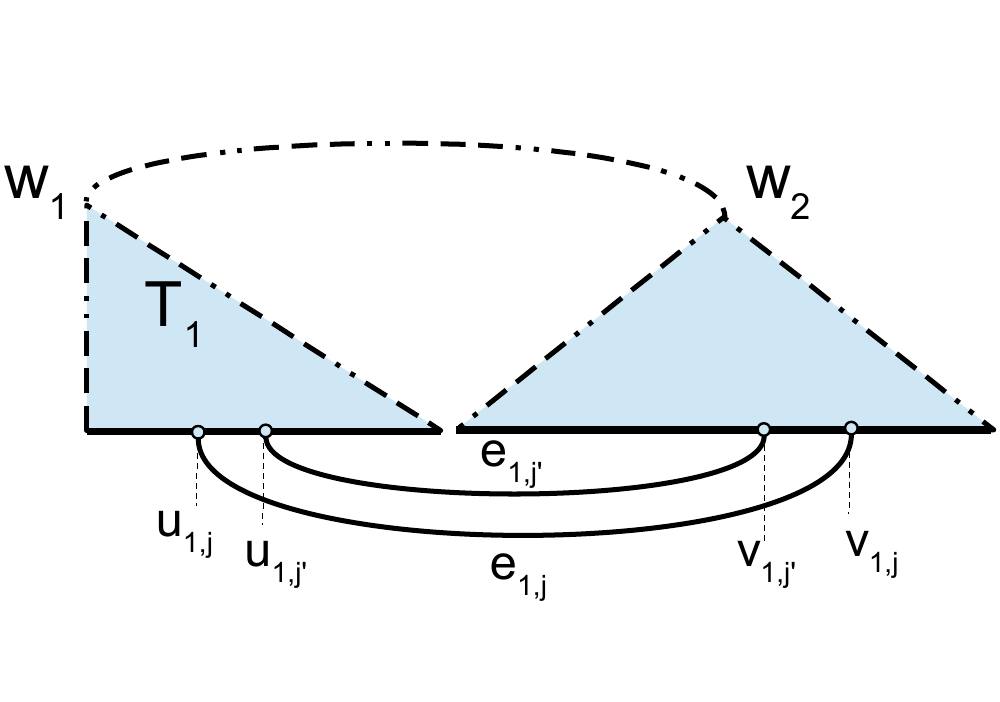}
                \caption{In the layout $\vfi^{\diamond}$ all the vertices of $T_1$ are labeled before the vertices of $H - T_1$.}
                \label{fig:first-subtree-labeled-separately-after}
        \end{subfigure}
        \caption{The OLA $\vfi^{*}$ and the alternative layout $\vfi^{\diamond}$.}
        \label{fig:first-subtree}
\end{figure}
In layout $\vfi^{\diamond}$, defined as it follows, all the vertices of $T_1$ are labeled with integers smaller than all the labels of vertices in $V(H) -V(T_1)$ by shifting them to the left while keeping their relative orders unchanged.
\[
 \forall v \in V, \vfi^{\diamond} =
    \begin{cases}
       \vfi^{*}(v) -  \delta_{\vfi^{*}} (-, v, V(H) - V(T_1)) &\quad \text{if } v \in V(T_1)\\
       \vfi^{*}(v) + \delta_{\vfi^{*}} (v, -, V(T_1)) &\quad \text{if } v \notin V(T_1)
     \end{cases}
\]
Going from layout $\vfi^{*}$ to $\vfi^{\diamond}$, the equation~\ref{eq:first-subtree-LA-change} can be inferred.
 \begin{align}
    LA(\vfi^{*},H) - LA(\vfi^{\diamond},H) = &  \label{eq:first-subtree-LA-change} \\
    & \Delta + \nonumber \\
    & \lambda(\Set{w_1,w_2}, \vfi^{*}) - \lambda(\Set{w_1,w_2}, \vfi^{\diamond}) + \nonumber \\
    & \lambda(e_{1,j}, \vfi^{*}) - \lambda(e_{1,j}, \vfi^{\diamond}) + \nonumber \\
    & \lambda(e_{1,j'}, \vfi^{*}) - \lambda(e_{1,j'}, \vfi^{\diamond}) + \nonumber
 \end{align}
 \noindent
 Where $\Delta$ is the increase in the value of linear arrangement due to overlapping vertices.
\footnote{Let $v \in V(T_1)$ be the vertex with largest label and $v' \in T(H) - V(T_1)$ with smallest label according to $\vfi^{*}$.
 $u \in V(H)$ is in overlapping area if $\vfi^{*}(v') < \vfi^{*}(u) < \vfi^{*}(v)$.}

\paragraph{Value of $\Delta$: }
We define $\Delta_1$ and $\Delta_2$ respectively as the number of vertices of
$V(T_1)$ and $V(H) - V(T_1)$ in the overlapping area. More specifically:
 \begin{align*}
    \Delta_1 = |\Set{v \in V(T_1) \text{ s.t. } \exists u, u' \in V(H) - V(T_1), \vfi^{*}(u) < \vfi^{*}(v) < \vfi^{*}(u')}| \\
    \Delta_2 = |\Set{v \in V(H) - V(T_1) \text{ s.t. } \exists u, u' \in V(T_1), \vfi^{*}(u) < \vfi^{*}(v) < \vfi^{*}(u')}| \\
 \end{align*}
\begin{fact}
    \label{fact:subtrees-two-or-three-connected}
    As presented in Figure~\ref{fig:first-subtree}, the set of vertices $V(T_1)$ and $V(H) - V(T_1)$ are connected via exactly three outgoing edges $\Set{w_1,w_2}$, $e_{1,j}$ and $e_{1,j'}$. Based on the three-connectivity of Halin graphs, any subset $\mathcal{V} \subset V(T_1)$ not incident to the outgoing edges, is connected to the rest of $V(T_1)$ by at least three edge disjoint paths. Also Any subset $\mathcal{V} \subset V(T_1)$ incident to some of outgoing edges $\Set{w_1,w_2}$, $e_{1,j}$ and $e_{1,j'}$, is connected to $V(T_1) - \mathcal{V}$ via at least two edge-disjoint paths. the same property hold for any $\mathcal{V} \subset V(H) - V(T_1)$.
\end{fact}
According to fact~\ref{fact:subtrees-two-or-three-connected}, any vertex $v \in V(T_1)$ in the overlapping area participates one unit in increasing the
expand of at least two edges of $E(H/T_1)$. Similarly any vertex $v \in V(H) - V(T_1)$ in overlapping area, increases the expand of at least two edges from $E(T_1)$.
Hence:
\begin{align}
    \Delta \ge 2 \times (\Delta_1 + \Delta_2) \label{eq:first-subtree-LA-change-Delta}
\end{align}
\paragraph{Change in the expands of $\Set{w_1,w_2}$, $e_{1,j}$ and $e_{2,j'}$: }
Based on the procedure that $\vfi^{\diamond}$ is constructed from $\vfi^{*}$ it's easy to validate the following equations.
 \begin{align}
   & (\lambda(\Set{w_1,w_2}, \vfi^{*}) - \lambda(\Set{w_1,w_2}, \vfi^{\diamond})) = - \delta(w_2,-,V(T_1)) \ge -\Delta_1 \label{eq:first-subtree-LA-change-w1w2}\\
   & (\lambda(e_{1,j}, \vfi^{*}) - \lambda(e_{1,j}, \vfi^{\diamond})) = - (\delta(-,u_{1,j},V(H)-V(T_1)) + \delta(v_{1,j},-,V(T_1)))  >  - (\Delta_2 + \Delta_1) \label{eq:first-subtree-LA-change-e1j}\\
   & (\lambda(e_{1,j'}, \vfi^{*}) - \lambda(e_{1,j'}, \vfi^{\diamond})) = - (\delta(-,u_{1,j'},V(H)-V(T_1)) + \delta(v_{1,j'},-,V(T_1)))  >  - (\Delta_2 + \Delta_1) \label{eq:first-subtree-LA-change-e1jprime}
\end{align}
\begin{remxxx}
\label{rem:canceled-out-sections}
Let $v \in \Set{w_2, v_{1,j}, v_{1,j'}}$ be the vertex with largest label among the three.
The rearrangement of $v$ increases the expand of the corresponding edges by $\delta(v,-,V(T_1)))$. But notice that based on fact~\ref{fact:subtrees-two-or-three-connected}
the set of vertices of $V(H) - V(T_1)$ labeled after $v$ are connected to rest of vertices (vertices on left side according to $\vfi^{*}$) using at least three vertices. Hence each vertex of $V(T_1)$ after $v$ (labeled with integers larger than label of $v$) add one unit to the expands of at least \underline{three} edges of $H/T_1$, while only the expands of two edges where considered in equation~\ref{eq:first-subtree-LA-change-Delta}.
Therefore the value $\delta(v,-,V(T_1))$ in the increase of the expand of the edge incident to $v$ must be ignored in the calculation of $ LA(\vfi^{*},H) - LA(\vfi^{\diamond},H)$.
\end{remxxx}
\noindent
Considering the remark ~\ref{rem:canceled-out-sections}, we finalize the proof by the following contradictory result.
 \begin{align}
    LA(\vfi^{*},H) - LA(\vfi^{\diamond},H) > & \\
    & 2 \times (\Delta_1 + \Delta_2) - \Delta_1 - (\Delta_2 + \Delta_1) - \Delta_2 \\
    & > 0
 \end{align}

\end{proof}
\begin{Prxxx}{(Proof of lemma~\ref{lemma:halin-subtrees-labeled-separately})}
\label{proof:halin-subtrees-labeled-separately}
In lemma~\ref{lemma:halin-first-subtree-labeled-separately},
it is shown that, given an OLA $\vfi^{*}$ for the Halin graph $H$, all the vertices of $T_1$
are labeled with continuous integers and hence are arranged before all other vertices in the graph.
Using a similar approach as in lemma~\ref{lemma:halin-first-subtree-labeled-separately},
we show that in an OLA $\vfi^{*}$ all the vertices in $V(T_2)$ are
labeled before all the vertices of $V(T_3) \cup \ldots \cup V(T_l)$.
Proof is complete as the same approach can be carried out to show that $ \forall 2 < i < l$ all the vertices of $V(T_i)$ are labeled with integers smaller than the labels of vertices of $V(T_{i+1}) \cup \ldots \cup V(T_l)$.

\noindent
In the rest of this proof we present $T_3 \uplus \ldots \uplus T_l$ by $\overline{T}_{1,2}$.
We show that a layout $\vfi^{*}$, where the arrangement of vertices of
$T_2$ overlap with the arrangement of vertices in $V(\overline{T}_{1,2})$ (as shown in Figure~\ref{fig:second-subtree-labeled-separately})
cannot be optimal. In order to do so, based on this allegedly optimal layout $\vfi^{*}$, we define the modified layout $\vfi^{\diamond}$
(presented by Figure~\ref{fig:second-subtree-labeled-separately-after}) and we show that $LA(\vfi^{*}) > LA(\vfi^{\diamond})$.
\begin{figure}
        \begin{subfigure}[b]{0.4\textwidth}
                \includegraphics[scale=.65]
                {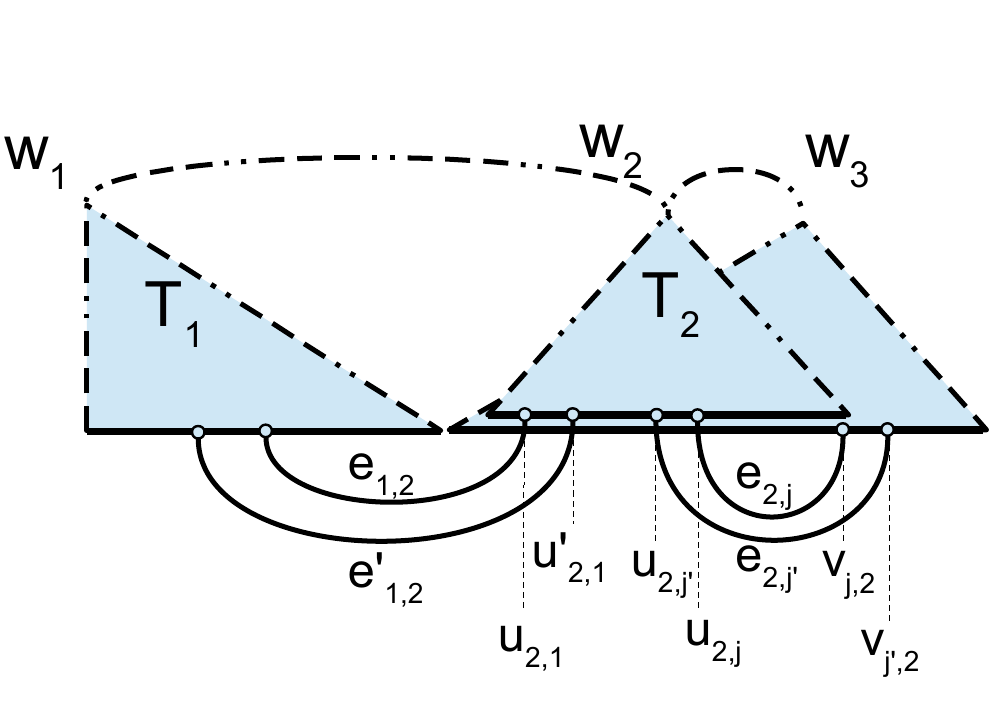}
                \caption{OLA layout $\vfi^{*}$ which labels some vertex $u \in V(T_2)$ with integers larger than the label of some vertex $v \in V(\overline{T}_{1,2})$. In other word the arrangement of $T_2$ has overlapping area with arrangement of $T_3 \uplus \ldots \uplus T_l$.}
                \label{fig:second-subtree-labeled-separately}
        \end{subfigure}
        \qquad\qquad
        \begin{subfigure}[b]{0.4\textwidth}
                \includegraphics[scale=.65]
                {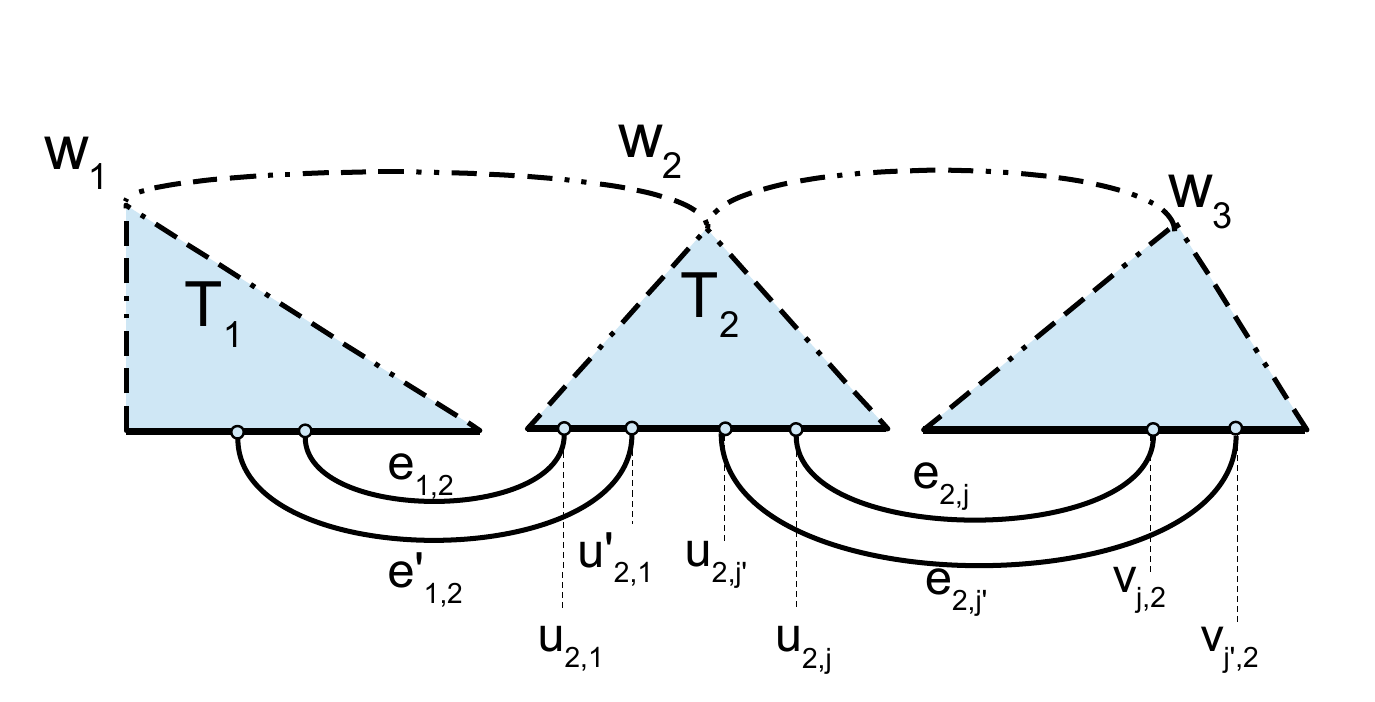}
                \caption{In the layout $\vfi^{\diamond}$ all the vertices of $T_2$ are labeled after the vertices of $T_1$ and before all the vertices of $V(\overline{T}_{1,2})$.}
                \label{fig:second-subtree-labeled-separately-after}
        \end{subfigure}
        \caption{The OLA $\vfi^{*}$ and the alternative layout $\vfi^{\diamond}$. Subtree $T_2$ is connected to $T_1$ via at most two edges $e_{1,2}, e ^{\prime} _{1,2} \in E(C)$,
        (respectively incident to $u_{1,2}, u^{\prime}_{1,2} \in V(T_2)$)
        and exactly one edge $\Set{w_1,w_2} \in E(T)$. $T_2$ is also connected to the rest of graph by the same number of edges0
        $e_{2,j}, e_{2,j'} \in E(C)$ and $\Set{w_2,w_3}$. End points of edge $e_{2,j}$ are $u_{2,j} \in V(T_2)$ and $v_{2,j} \in V(\overline{T}_{1,2})$ while $u_{2,j'} \in V(T_2)$ and $v_{2,j'} \in V(\overline{T}_{1,2})$ are the two end points of $e_{2,j'}$.}
        \label{fig:second-subtree}
\end{figure}

\noindent
The subtree $T_2$ is connected to $T_1$ through one edge $e_{1,2} \in E(T)$ and one or two edges from $E(C)$. $T_2$ is also connected to the rest of graph by exactly the same number of edges. As presented in Figure~\ref{fig:second-subtree}, we only consider the case where $T_2$
is connected to each of subgraphs $T_2$ and $T_3 \uplus \ldots \uplus T_l$ by one edge of $E(T)$ and two edges of $E(C)$. The other case can be analyze in the same way and is omitted.

\noindent
Layout $\vfi^{\diamond}$ is formally defined as it follows.
\[
 \forall v \in V, \vfi^{\diamond} =
    \begin{cases}
       \vfi^{*}(v) -  \delta_{\vfi^{*}} (-, v, V(\overline{T}_{1,2})) &\quad \text{if } v \in V(T_2)\\
       \vfi^{*}(v) + \delta_{\vfi^{*}} (v, -, V(T_2)) &\quad \text{if } v \in V(\overline{T}_{1,2})
     \end{cases}
\]

\noindent
Based on this definition it's easy to see that:
 \begin{align}
    LA(\vfi^{*},H) - LA(\vfi^{\diamond},H) = &  \label{eq:second-subtree-LA-change} \\
    & \Delta + \nonumber \\
    & (\lambda(\Set{w_1,w_2}, \vfi^{*}) - \lambda(\Set{w_1,w_2}, \vfi^{\diamond})) + \nonumber \\
    & (\lambda(\Set{w_2,w_3}, \vfi^{*}) - \lambda(\Set{w_2,w_3}, \vfi^{\diamond})) + \nonumber \\
    & (\lambda(e_{1,2}, \vfi^{*}) - \lambda(e_{1,2}, \vfi^{\diamond})) + \nonumber \\
    & (\lambda(e^{\prime} _{1,2}, \vfi^{*}) - \lambda(e^{\prime} _{1,2}, \vfi^{\diamond})) + \nonumber \\
    & (\lambda(e_{2,j} , \vfi^{*}) - \lambda(e_{2,j}, \vfi^{\diamond})) + \nonumber \\
    & (\lambda(e_{2,j'}, \vfi^{*}) - \lambda(e_{2,j'}, \vfi^{\diamond})) + \nonumber
 \end{align}
As in lemma~\ref{lemma:halin-first-subtree-labeled-separately} the increase in the value of linear arrangement due to overlap
 is presented by $\Delta$ and we define $\Delta_2$ and $\Delta_3$ as:
  \begin{align*}
    \Delta_2 = |\Set{v \in V(T_2) \text{ s.t. } \exists u, u' \in V(\overline{T}_{1,2}), \vfi^{*}(u) < \vfi^{*}(v) < \vfi^{*}(u')}| \\
    \Delta_3 = |\Set{v \in V(\overline{T}_{1,2}) \text{ s.t. } \exists u, u' \in V(T_2), \vfi^{*}(u) < \vfi^{*}(v) < \vfi^{*}(u')}| \\
 \end{align*}
 Hence $\Delta_2$ and $\Delta_3$ respectively correspond to the number of vertices of $T_2$ and $T_3 \uplus \ldots \uplus T_l$ which are
 in the overlapping area based on $\vfi^{*}$.

  \begin{fact}
    \label{fact:second-subtree-two-or-three-connected}
    Similar to the fact~\ref{fact:subtrees-two-or-three-connected} and according to the three-connectivity of Halin graphs, any subset $\mathcal{V} \subset V(\overline{T}_{1,2})$ not incident to the outgoing edges $\Set{w_2,w_3}$, $e_{2,j}$ and $e_{2,j'}$, is connected to the rest of $V(\overline{T}_{1,2})$ through at least three edge disjoint paths. Also Any subset $\mathcal{V} \subset V(\overline{T}_{1,2})$ incident to some of the outgoing edges, is connected to $V(\overline{T}_{1,2}) - \mathcal{V}$ via at least two edge-disjoint paths. The same property hold for any $\mathcal{V} \subset V(T_1) \cup V(T_2)$.
\end{fact}
\paragraph{Value of $\Delta$: }
Using the fact~\ref{fact:second-subtree-two-or-three-connected} one can see that each vertex in the overlapping area, increases the expand of at least two edges by one unit. Hence the following equation can be deduced.
\begin{align}
    \Delta \ge 2 \times (\Delta_2 + \Delta_3) \label{eq:second-subtree-LA-change-Delta}
\end{align}
\paragraph{Change in the expands of outgoing edges:}
As in lemma~\ref{lemma:halin-first-subtree-labeled-separately}, the change in the expand of edges linking $T_2$ to the rest of graph
can be derived as:
 \begin{align}
   & (\lambda(\Set{w_1,w_2}, \vfi^{*}) - \lambda(\Set{w_1,w_2}, \vfi^{\diamond})) = \delta(-, w_2, V(\overline{T}_{1,2}))
   \label{eq:second-subtree-LA-change-w2w1}\\
    & (\lambda(\Set{w_2,w_3}, \vfi^{*}) - \lambda(\Set{w_2,w_3}, \vfi^{\diamond})) = - (\delta(-, w_2, V(\overline{T}_{1,2})) +
    \delta(w_3, -,V(T_2)) )
    \label{eq:second-subtree-LA-change-w2w3}\\
   & (\lambda(e_{1,2}, \vfi^{*}) - \lambda(e_{1,2}, \vfi^{\diamond})) = \delta(-,u_{1,2},V(\overline{T}_{1,2}))
   \label{eq:second-subtree-LA-change-e12}\\
   & (\lambda(e^{\prime} _{1,2}, \vfi^{*}) - \lambda(e^{\prime} _{1,2}, \vfi^{\diamond})) = \delta(-,u^{\prime} _{1,2},V(\overline{T}_{1,2}))
   \label{eq:second-subtree-LA-change-e12prime}\\
   & (\lambda(e_{2,j}, \vfi^{*}) - \lambda(e_{2,j}, \vfi^{\diamond})) \ge - (\delta(-,u_{2,j},V(\overline{T}_{1,2})) + \delta(v_{j,2},-,V(T_2)))
   \label{eq:second-subtree-LA-change-e2j} \\
      & (\lambda(e_{2,j'}, \vfi^{*}) - \lambda(e_{2,j'}, \vfi^{\diamond})) \ge - (\delta(-,u_{2,j'},V(\overline{T}_{1,2})) + \delta(v_{j',2},-,V(T_2)))
   \label{eq:second-subtree-LA-change-e2jprime}
\end{align}
\begin{remxxx}
\label{rem:canceled-out-sections}
Let $v \in \Set{w_3, v_{j,2}, v_{j',2}}$ be the vertex with largest label among the three. According to the fact~\ref{fact:second-subtree-two-or-three-connected} and using the same
reasoning as in remark~\ref{rem:canceled-out-sections}, each vertex in $V(T_2)$ labeled after $v$ takes part in the increase of expand of at least \underline{three} edges of $V(\overline{T}_{1,2})$, while only two where considered in the calculation of $\Delta$ in equation~\ref{eq:second-subtree-LA-change-Delta}.
Hence the value $\delta(v,-,V(T_2))$, considered for the change in expand of  the edge incident to $v$, should be added back to the calculation of $LA(\vfi^{*},H) - LA(\vfi^{\diamond},H)$. Without loss of generality in the rest of the proof we assume $v = v_{j',2}$.
\end{remxxx}

\noindent
Accordingly equation~\ref{eq:second-subtree-LA-change} can be simplified as it follows.
 \begin{align}
    LA(\vfi^{*},H) - LA(\vfi^{\diamond},H) \ge & \label{eq:second-subtree-LA-change2}\\
    & 2 \times (\Delta_2 + \Delta_3) \nonumber \\
    & - \delta(w_3, -,V(T_2)) \nonumber \\
    & - \delta(u_{1,2}, u_{2,j},V(\overline{T}_{1,2})) \nonumber \\
    & - \delta(u^{\prime}_{1,2}, u_{2,j'},V(\overline{T}_{1,2})) \nonumber \\
    & - \delta(v_{j,2},-,V(T_{1,2}))
 \end{align}
 \noindent
 \begin{remxxx}
Depending on the order of labels of $u_{1,2}$ and $u_{2,j}$, $\delta(-, u_{1,2},V(\overline{T}_{1,2})) - \delta(-, u_{2,j},V(\overline{T}_{1,2}))$
is either equal to $\delta(u_{2,j}, u_{1,2},V(\overline{T}_{1,2}))$ or $- \delta(u_{1,2}, u_{2,j},V(\overline{T}_{1,2}))$.
The same way we can reason about $\delta(-, u^{\prime}_{1,2},V(\overline{T}_{1,2})) - \delta(-, u_{2,j'},V(\overline{T}_{1,2}))$ as following:
 \begin{align}
     \delta(-, u_{1,2},V(\overline{T}_{1,2})) - \delta(-, u_{2,j},V(\overline{T}_{1,2})) \ge &
     - \delta(u_{1,2}, u_{2,j},V(\overline{T}_{1,2})) \ge -\Delta_3  \label{eq:second-subtree-LA-change-u12-22j}\\
    \delta(-, u^{\prime}_{1,2},V(\overline{T}_{1,2})) - \delta(-, u_{2,j'},V(\overline{T}_{1,2})) \ge &
    - \delta(u^{\prime}_{1,2}, u_{2,j'},V(\overline{T}_{1,2})) \ge -\Delta_3 \label{eq:second-subtree-LA-change-uprime12-22jprime}
 \end{align}
Also it is easy to see that:
 \begin{align}
    & - \delta(w_3, -,V(T_2)) \ge \Delta_2 \\
    & - \delta(v_{2,j},-,V(T_{1,2})) \ge \Delta_2 \label{eq:second-subtree-LA-change-v2j}
 \end{align}
 \end{remxxx}

\begin{remxxx}
    The equality in equation~\ref{eq:second-subtree-LA-change-e2j} holds only if $u_{2,j}$ is labeled before $v_{j,2}$ (namely $\vfi^{*}(u_{2,j}) < \vfi^{*}(v_{2,j})$)
    \footnote{Similarly the equality in equation~\ref{eq:second-subtree-LA-change-uprime12-22jprime} hold only if $\vfi^{*}(u^{\prime}_{2,j}) < \vfi^{*}(v_{2,j'})$.}.
    But the equalities in equations~\ref{eq:second-subtree-LA-change-u12-22j} and~\ref{eq:second-subtree-LA-change-v2j} can hold
    simultaneously only if $v_{2,j}$ is labeled before all the vertices $V(T_2)$ in the overlapping area and $u_{2,j}$ is labeled after
    all the vertices of $\overline{T}_{1,2}$ which are in the overlapping area. In other word, both the equalities in equations\ref{eq:second-subtree-LA-change-u12-22j} and~\ref{eq:second-subtree-LA-change-v2j} hold only if
    $\vfi^{*}(u_{2,j}) > \vfi^{*}(v_{2,j})$. Accordingly the equalities in equations~\ref{eq:second-subtree-LA-change-e2j},~\ref{eq:second-subtree-LA-change-u12-22j} and~\ref{eq:second-subtree-LA-change-v2j} never simultaneously hold.
\end{remxxx}

\noindent
Consequently the equation~\ref{eq:second-subtree-LA-change2} can be simplified as following with the contradictory result that completes the proof.
 \begin{align}
    LA(\vfi^{*},H) - LA(\vfi^{\diamond},H)  > & \\
    & 2 \times (\Delta_2 + \Delta_3) \nonumber \\
    & - \Delta_2 \nonumber \\
    & - \Delta_3 \nonumber \\
    & - \Delta_3 \nonumber \\
    & - \Delta_2 = 0
 \end{align}

\end{Prxxx}
\begin{Prxxx}{(Proof of lemma~\ref{lemma:halin-branches-labeled-separately})}
\label{proof:halin-branches-labeled-separately}
Each spinal branch $B_{i,j}$ is connected to the rest of the graph using three edges. $e_{j,w} = \Set{v_j, w_i} \in E(T)$ connecting it to the spinal vertex $w_i$ and two edge $e_{j,j'} \in E(C)$ and $e_{j,j''} \in E(C)$ connecting $B_{i,j}$ to two other branches $B_{i',j'}$ and $B_{i'',j''}$ \footnote{Note that based on the structure of Halin graphs, $B_{i',j'}$ and $B_{i'',j''}$ may belong to the same subtree $T_i$ as $B_{i,j}$, but both can not be a part of the same subtree $T_i'$ different from $T_i$}.

\noindent
Without loss of generality we assume $j = 1$ and $j' = 2$ and we consider two branches $B_{i,1}$ and $B_{i,2}$ are
anchored at $v_1$ and $v_2$, and we only show the following for the case where $\vfi^{*}$ is of type
$(\ldots,w_i,\ldots,V(B_{i,1}) \cup V(B_{i,2}),\ldots)$.
 \begin{align*}
    \forall v \in V(B_{i,1}) \cup V(B_{i,2}); \vfi^{*}(w_i) < \vfi^{*}(v)  \Rightarrow \\
    & (\forall v \in V(B_{i,1}), u \in V(B_{i,2}): \vfi^{*}(v) < \vfi^{*}(u)) \vee \\
    & (\forall v \in V(B_{i,1}), u \in V(B_{i,2}): \vfi^{*}(v) > \vfi^{*}(u))
 \end{align*}

\noindent
Assume two vertices $\underline{v}, \overline{v} \in V(B_{i,1}) \cup V(B_{i,2})$ such that
$\forall u \in V(B_{i,1}) \cup V(B_{i,2}), \vfi^{*}(\underline{v}) \le u \le \vfi^{*}(\overline{v})$.

\noindent
\paragraph{Case 1: $\underline{v} \in B_{i,1}$ and $\overline{v} \in B_{i,2}$.}
\label{branches-labeled-separately-case1}
There are two possible sub-cases: 1.1) there is no edge connecting $B_{i,1}$ and $B_{i,2}$. 1.2) $B_{i,1}$ is connected to $B_{i,2}$ using exactly one edge $e_{1,2} \in E(C)$ (Figure~\ref{fig:branches-overlap-gen}). We present the proof for the later case. The analysis of proof of former case is similar and omitted. We show that the assumption of lemma being false, in other word if two branches of $B_{i,1}$ and $B_{i,1}$ \emph{overlap}, contradicts the optimality assumption of $\vfi^{*}$.
Accordingly for an overlapping OLA $\vfi^{*}$, we present an alternative layout $\vfi^{\diamond}_L$ and finish the proof by showing
the contradictory result $LA(\vfi^{*}) > LA(\vfi^{\diamond}_L)$.
\begin{figure}
  \begin{center}
      \includegraphics[scale=.6]
                {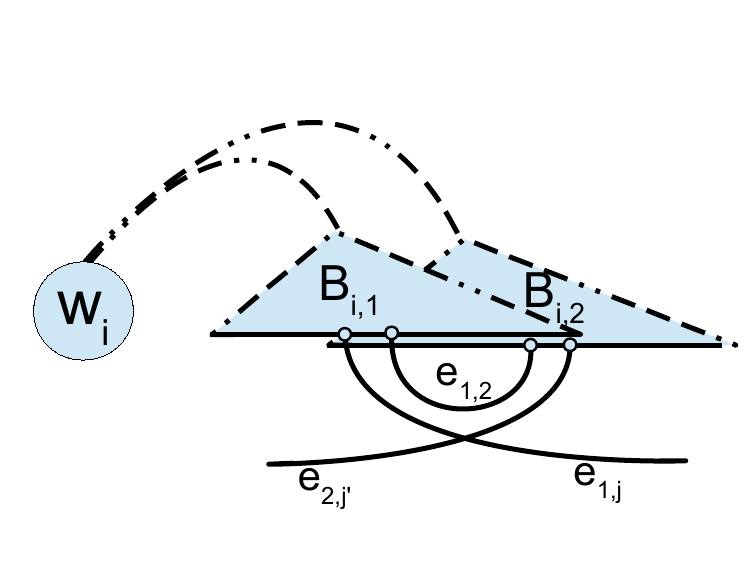}
  \end{center}
 \caption{An branch overlapping layout where the two branches $B_{i,1}$ and $B_{i,2}$ are connected.}
 \label{fig:branches-overlap-gen}
\end{figure}
\begin{figure}
        \centering
        \begin{subfigure}[b]{0.3\textwidth}
                \includegraphics[scale=.6]
                {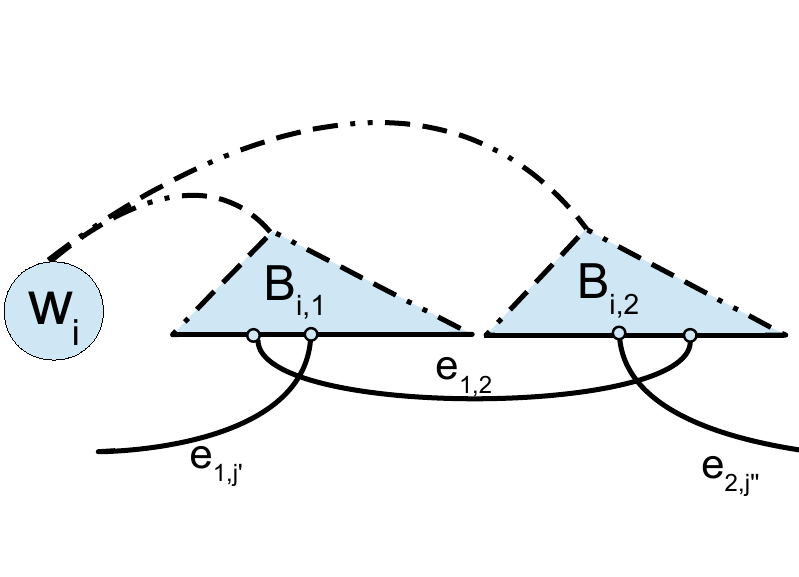}
                \caption{The alternative non-overlapping layout $\vfi^{\diamond}_L$.}
                \label{fig:branches-overlap-afterL}
        \end{subfigure}
        \qquad\qquad\qquad
        \begin{subfigure}[b]{0.3\textwidth}
                \includegraphics[scale=.6]
                {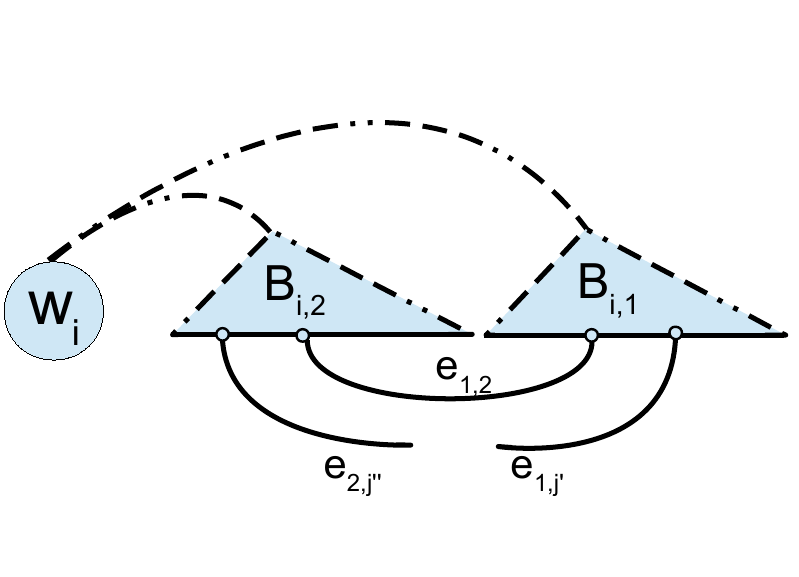}
                \caption{The alternative non-overlapping layout $\vfi^{\diamond}_R$.}
                \label{fig:branches-overlap-afterR}
        \end{subfigure}
        \caption{In the layout $\vfi^{\diamond}_L$ all the vertices of $B_{i,1}$ are labeled on the left side of vertices of $B_{i,2}$
        and on the right side in the layout $\vfi^{\diamond}_R$.}
        \label{fig:subtees-layouts}
\end{figure}
In the alternative layout $\vfi^{\diamond}_L$ all the vertices of $B_{i,1}$ are labeled before all the vertices of $B_{i,2}$
while the relative order of labels of other vertices are preserved the same. Formally $\vfi^{\diamond}_L$ is defined as follows.
\[
 \forall v \in V, \vfi^{\diamond}_L =
    \begin{cases}
       \vfi^{*}(v) &\quad \text{if } v \notin V(B_{i,1}) \cup V(B_{i,1})\\
       \vfi^{*}(v) -  \delta_{\vfi^{*}} (-, v, V(B_{i,2})) &\quad \text{if } v \in V(B_{i,1})\\
       \vfi^{*}(v) + \delta_{\vfi^{*}} (v, -, V(B_{i,1})) &\quad \text{if } v \in V(B_{i,2})\
     \end{cases}
\]
Based on the definition of  $\vfi^{\diamond}_L$ and from Figure~\ref{fig:branches-overlap-afterL} the following holds:
 \begin{align}
    LA(\vfi^{*}) - LA(\vfi^{\diamond}_L) = & \\
    & \Delta + \nonumber \\
    & (\lambda(e_{1,2}, \vfi^{*}) - \lambda(e_{1,2}, \vfi^{\diamond}_L)) + \nonumber \\
    & (\lambda(e_{1,j'}, \vfi^{*}) - \lambda(e_{1,j'}, \vfi^{\diamond}_L)) + \nonumber \\
    & (\lambda(e_{2,j}, \vfi^{*}) - \lambda(e_{2,j}, \vfi^{\diamond}_L)) + \nonumber \\
    & (\lambda(e_{1,w}, \vfi^{*}) - \lambda(e_{1,w}, \vfi^{\diamond}_L)) + \nonumber \\
    & (\lambda(e_{2,w}, \vfi^{*}) - \lambda(e_{1,w}, \vfi^{\diamond}_L)) \nonumber
 \end{align}
 \noindent
 As before, $\Delta$ represents the increase in the value of linear arrangement due to overlap.

 \noindent
 Let $\Delta_1$ and $\Delta_2$ respectively be the number of vertices of
$B_{i,1}$ and $B_{i,2}$ in the overlapping area. Formally speaking:
 \begin{align*}
    \Delta_1 = |\Set{v \in B_{i,1} | \exists u, u' \in B_{i,2}, \vfi^{*}(u) < \vfi^{*}(v) < \vfi^{*}(u')}| \\
    \Delta_2 = |\Set{v \in B_{i,2} | \exists u, u' \in B_{i,1}, \vfi^{*}(u) < \vfi^{*}(v) < \vfi^{*}(u')}| \\
 \end{align*}
\begin{fact}
\label{subbranch-three-connected-part}
Every spinal branch $B_{i,j}$ of a Halin graph $H = T \uplus C$, anchored at $v_{i,j}$, is connected to the rest of $H$ via three outgoing edges,
$\Set{w_i, v_{i,j}} \in E(T)$, $e_{j,j'} \in E(C)$ and $e_{j,j''} \in E(C)$. The two edges $e_{j,j'}$ and $e_{j,j''}$, respectively incident to $v_R$ and $v_L$, connect $B_{i,j}$ to two other spinal branches $B_{i',j'}$ and $B_{i'',j''}$
\footnote{$e_{j,j'}$ and $e_{j,j''}$ are the right and left outgoing edges of $B_{i,j}$}.
Every vertex of a non-empty \emph{sub-branch} $\mathcal{B} \subset B_{i,j} - \Set{v_{i,j}, v_{j,R}, v_{j,L}}$ is connected to $B_{i,j}/ \mathcal{B}$ via at least three edge disjoint paths. Consequently $\mathcal{B}$ is connected to $B_{i,j}/ \mathcal{B}$ by at least three edges.

\noindent
In general every vertex of a non-empty set $\mathcal{B} \subset B_{i,j}$ (which may contains some of the vertices of $\Set{v_{i,j}, v_{j,R}, v_{j,L}}$) is connected to $B_{i,j}/ \mathcal{B}$ by at least two edges-disjoint paths.
\end{fact}
\paragraph{Value of $\Delta$: }
As a consequence of fact~\ref{subbranch-three-connected-part},
each vertex of a branch in the overlapping area contribute one unit to the increase in the expand of at least two edges from the other branch.
Hence it is the case that:
\begin{align}
    \Delta \ge 2 \times (\Delta_1 + \Delta_2) \label{eq:overlap-penalty-case1}
\end{align}
\paragraph{Change in the expands of $e_{1,2}$, $e_{1,j'}$, $e_{2,j''}$ and $e_{2,w}$: }
The increase in the expand of each of these edges is equivalent to how much the two end points drift apart in construction of $\vfi^{\diamond}_L$.
Accordingly the following equations are easy to verify:
 \begin{align}
   & (\lambda(e_{1,2}, \vfi^{*}) - \lambda(e_{1,2}, \vfi^{\diamond}_L)) \ge  - (\Delta_1 + \Delta_2)  \label{eq:expands-outergoing-edges-1} \\
   & (\lambda(e_{1,j'}, \vfi^{*}) - \lambda(e_{1,j'}, \vfi^{\diamond}_L)) \ge  - \Delta_2 \label{eq:expands-outergoing-edges-2}\\
   & (\lambda(e_{2,j''}, \vfi^{*}) - \lambda(e_{2,j''}, \vfi^{\diamond}_L)) \ge  - \Delta_1 \label{eq:expands-outergoing-edges-3} \\
   & (\lambda(e_{2,w}, \vfi^{*}) - \lambda(e_{2,w}, \vfi^{\diamond}_L)) \ge  - \Delta_1 \label{eq:expands-outergoing-edges-4}
\end{align}
\begin{remxxx}
\label{rem:after-last-edge-three-connected}
Let $u \in V(B_{i,2})$ be the vertex incident to one of the edges $e_{1,2}$, $e_{2,j''}$ and $e_{2,w}$ with the largest label based on $\vfi^{*}$.
Using fact~\ref{subbranch-three-connected-part}, the set of vertices of $\overrightarrow{\mathcal{V}}_u \subset V(B_{i,2})$ labeled after $u$ is three-connected to set of vertices labeled before $u$. Accordingly each vertex of $B_{i,1}$ labeled with an integer larger that $\vfi^{*}(u)$ contributes one unit to the increase in expand of at least \underline{three} edges of $B_{i,2}$. Therefore this value cancels out the expand of the edge incident
to $u$.
\end{remxxx}
\begin{remxxx}
\label{rem:not-all-the-inequalities-hold}
It's easy to see that at most only one of the equalities~\ref{eq:expands-outergoing-edges-1} to~\ref{eq:expands-outergoing-edges-4} can hold.
\end{remxxx}
\noindent
The following contradictory result, from putting the equations~\ref{eq:overlap-penalty-case1} to \ref{eq:expands-outergoing-edges-3} and remarks~\ref{rem:after-last-edge-three-connected} and~\ref{rem:not-all-the-inequalities-hold} together, concludes our proof in this case.
 \begin{align*}
    LA(\vfi^{*}) - LA(\vfi^{\diamond}_L) >
    2 \times (\Delta_1 + \Delta_2) +
    - (\Delta_1 + \Delta_2) -
    \Delta_1 -
    \Delta_2
    = 0
 \end{align*}

\noindent
\paragraph{Case 2: $\underline{v} \in B_{i,2}$ and $\overline{v} \in B_{i,1}$.} This case is symmetric to the previous case and in the alternative layout $\vfi^{\diamond}_L$ all the vertices of $B_{i,1}$ are labeled after those of $B_{i,2}$. Hence the proof is similar and is omitted.

\noindent
\paragraph{Case 3: $\underline{v}, \overline{v} \in B_{i,1}$.} Layout $\vfi^{*}$ for this case is presented in
Figure~\ref{fig:branches-overlap-case3} where all the vertices of $B_{i,2}$ are enclosed by $B_{i,1}$.
In contrast to layout $\vfi^{*}$, we present layout $\vfi^{\diamond}_R$  where all vertices of $B_{i,1}$ are labeled
on the right side of those of $B_{i,2}$ as formally defined in following.
\[
 \forall v \in V, \vfi^{\diamond}_R =
    \begin{cases}
       \vfi^{*}(v) &\quad \text{if } v \notin V(B_{i,1}) \cup V(B_{i,1})\\
       \vfi^{*}(v) +  \delta_{\vfi^{*}} (v, \overline{v}, V(B_{i,2})) &\quad \text{if } v \in V(B_{i,1})\\
       \vfi^{*}(v) - \delta_{\vfi^{*}} (w_i, v, V(B_{i,1})) &\quad \text{if } v \in V(B_{i,2})\
     \end{cases}
\]
\noindent
We finish the proof by showing that either $LA(\vfi^{*}) > \vfi^{\diamond}_R$ or $LA(\vfi^{*}) > \vfi^{\diamond}_L$.
\begin{figure}
  \begin{center}
     \includegraphics [scale=.7] %[width=0.35\textwidth,height=5cm]%
          {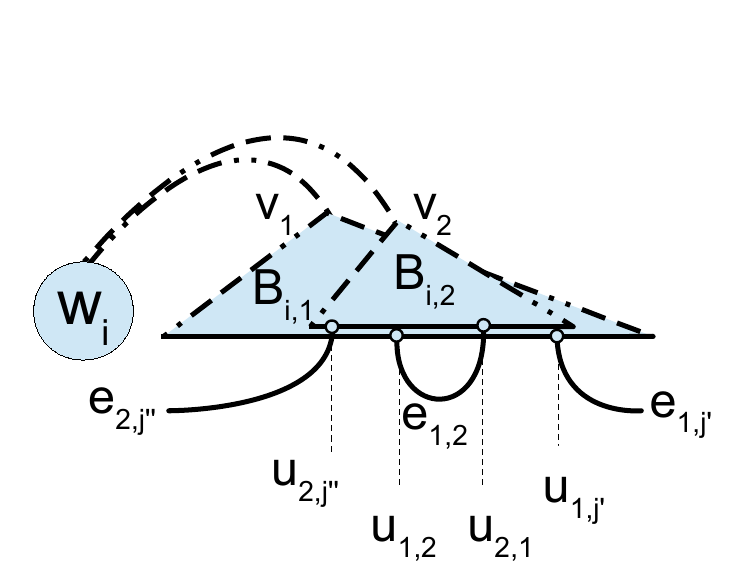}
  \end{center}
     \caption{General presentation of an OLA $\vfi^{*}$ where $\underline{v}, \overline{v} \in B_{i,1}$.
     $u_{1,2} \in V(B_{i,1})$ and $u_{2,1} \in V(B_{i,2})$ are the two end points of edge $e_{1,2}$ while
     $u_{1,j'} \in V(B_{i,1})$ and $u_{2,j''} \in V(B_{i,2})$ are respectively one of the end points of two edge
     $e_{1,j'}$ and $e_{2,j''}$.
     }
   \label{fig:branches-overlap-case3}
\end{figure}
According to the definitions of $\vfi^{\diamond}_R$ and $\vfi^{\diamond}_L$ one can inferred the
equation~\ref{eq:branches-enclave-after-L-R}.
 \begin{align}
 \label{eq:branches-enclave-after-L-R}
    LA(\vfi^{*}) - LA(\vfi^{\diamond}_{-}) = & \\
    & \Delta + \nonumber \\
    & (\lambda(e_{1,2}, \vfi^{*}) - \lambda(e_{1,2}, \vfi^{\diamond}_{-})) + \nonumber \\
    & (\lambda(e_{1,j'}, \vfi^{*}) - \lambda(e_{1,j'}, \vfi^{\diamond}_{-})) + \nonumber \\
    & (\lambda(e_{2,j}, \vfi^{*}) - \lambda(e_{2,j}, \vfi^{\diamond}_{-})) + \nonumber \\
    & (\lambda(e_{1,w}, \vfi^{*}) - \lambda(e_{1,w}, \vfi^{\diamond}_{-})) + \nonumber \\
    & (\lambda(e_{2,w}, \vfi^{*}) - \lambda(e_{1,w}, \vfi^{\diamond}_{-})) \nonumber
 \end{align}
\noindent
Where the wildcard "-" can be replaced by $L$ or $R$, and as before $\Delta$ represents
the increase in the value of linear arrangement due to overlap.
\paragraph{Value of $\Delta$: } Considering the same definition for $\Delta_1 $ and $ \Delta_2$, then
$\Delta \ge 2 \times (\Delta_1 + \Delta_2)$. Since $\Delta_2 = \beta_{i,2}$ then:
 \begin{align}
   \Delta \ge 2 \times (\Delta_1 + \beta_{i,2}) \label{eq:overlap-penalty-case3}
 \end{align}
 \paragraph{Change in the expands of $e_{1,2}$, $e_{1,j'}$, $e_{2,j''}$, $e_{1,w}$ and $e_{2,w}$: }
 In the calculation of the change in expand an edge, in should be considered if the two end points are
 drifting apart or getting closer. Hence, noting the opposing definitions of $\vfi^{\diamond}_L$ and $\vfi^{\diamond}_R$,
 following equations hold.
  \begin{align}
    & (\lambda(e_{1,2}, \vfi^{*}) - \lambda(e_{1,2}, \vfi^{\diamond}_{L})) \ge  - (\beta_2 + \delta_{\vfi^{*}} (u_{2,1}, -, V(B_{i,1}))) \\
    & (\lambda(e_{1,j'}, \vfi^{*}) - \lambda(e_{1,j'}, \vfi^{\diamond}_{L})) =  - \alpha(e_{1,j'}) \times   \delta_{\vfi^{*}} (-, u_{1,j'}, V(B_{i,2}))\\
    & (\lambda(e_{2,j''}, \vfi^{*}) - \lambda(e_{2,j''}, \vfi^{\diamond}_{L})) = - \alpha(e_{2,j''}) \times   \delta_{\vfi^{*}} (u_{2,j''}, -, V(B_{i,1})) \\
    & (\lambda(e_{1,w}, \vfi^{*}) - \lambda(e_{1,w}, \vfi^{\diamond}_{L})) =   \delta_{\vfi^{*}} (-, v_1, V(B_{i,2})) \\
    & (\lambda(e_{2,w}, \vfi^{*}) - \lambda(e_{2,w}, \vfi^{\diamond}_{L})) = - \delta_{\vfi^{*}} (v_2, -, V(B_{i,1}))
    \label{eq:case3-e2w-L}
\end{align}

  \begin{align}
    & (\lambda(e_{1,2}, \vfi^{*}) - \lambda(e_{1,2}, \vfi^{\diamond}_{R})) \ge  - (\beta_2 + \delta_{\vfi^{*}} (-, u_{2,1}, V(B_{i,1})))
    \label{eq:case3-e12-R}\\
    & (\lambda(e_{1,j'}, \vfi^{*}) - \lambda(e_{1,j'}, \vfi^{\diamond}_{R})) =  \alpha(e_{1,j'}) \times   \delta_{\vfi^{*}} (u_{1,j'},-, V(B_{i,2})) \\
    & (\lambda(e_{2,j''}, \vfi^{*}) - \lambda(e_{2,j''}, \vfi^{\diamond}_{R})) = \alpha(e_{2,j''}) \times   \delta_{\vfi^{*}} (-, u_{2,j''}, V(B_{i,1})) \label{eq:case3-e2j-R}\\
    & (\lambda(e_{1,w}, \vfi^{*}) - \lambda(e_{1,w}, \vfi^{\diamond}_{R})) =   - \delta_{\vfi^{*}} (v_1,-, V(B_{i,2})) \ge -\beta_{i,2}
    \label{eq:case3-e1w-R} \\
    & (\lambda(e_{2,w}, \vfi^{*}) - \lambda(e_{2,w}, \vfi^{\diamond}_{R})) = \delta_{\vfi^{*}} (-, v_2, V(B_{i,1}))
\end{align}
The coefficient $\alpha(e)$ is $1$ if the edge $e$ is stretching, $-1$ if it's expand is decreasing and $0$ otherwise.
For instance if the expand of edge $e_{1,j'}$ increases based on $\vfi^{\diamond}_{L}$, it obviously will decrease based on $\vfi^{\diamond}_{R}$.
We break the rest of the proof to different sub-cases according to the signs of $\alpha(e_{1,j'})$ and $\alpha(e_{2,j''})$.
\paragraph{Case 3.1: $\alpha(e_{1,j'}) = 1$ and $\alpha(e_{2,j''}) = 1$.}
Therefore both edges $e_{1,j'}$ and $e_{2,j''}$ shrink based on $\vfi^{\diamond}_{R}$. Putting equations~\ref{eq:overlap-penalty-case3},~\ref{eq:case3-e12-R},~\ref{eq:case3-e2j-R} and~\ref{eq:case3-e1w-R} together, we conclude:
  \begin{align*}
   LA(\vfi^{*}) - LA(\vfi^{\diamond}_{R}) \ge & \\
   & 2 \times (\Delta_1 + \beta_{i,2}) \\
   & - \beta_2 - \delta_{\vfi^{*}} (-, u_{2,1}, V(B_{i,1})) \\
   & + \delta_{\vfi^{*}} (-, u_{2,j''}, V(B_{i,1}))\\
   & -\beta_{i,2} \ge 2 \times \Delta_1  + (\delta_{\vfi^{*}} (-, u_{2,j''}, V(B_{i,1})) - \delta_{\vfi^{*}} (-, u_{2,1}, V(B_{i,1})))  \ge \Delta_1
  \end{align*}

\noindent
Notice that $(\delta_{\vfi^{*}} (-, u_{2,j''}, V(B_{i,1})) - \delta_{\vfi^{*}} (-, u_{2,1}, V(B_{i,1}))) = \delta_{\vfi^{*}} (u_{2,1}, u_{2,1}, V(B_{i,1})) \ge - \Delta_1$.
Also if $\Delta_1 = 0$ \footnote{Layout $\vfi^{*}$ labels all the vertices of $B_{i,2}$ with continuous integers.}, then the equality~\ref{eq:case3-e12-R} cannot hold
\footnote{Due to the fact that $u_{1,2}$ is labeled after all the vertices of $B_{i,2}$, hence:
\[
\lambda(e_{1,2}, \vfi^{*}) - \lambda(e_{1,2}, \vfi^{\diamond}_{R})) \ge  - (\beta_2 + \delta_{\vfi^{*}} (-, u_{2,1}, V(B_{i,1})) - 1
\]
}.
Accordingly it is always the case that $LA(\vfi^{*}) - LA(\vfi^{\diamond}_{R}) > 0$.
\paragraph{Case 3.2: $\alpha(e_{1,j'}) = -1$ and $\alpha(e_{2,j''}) = -1$.}
Thus the expands of both edges $e_{1,j'}$ and $e_{2,j''}$ decrease going from $\vfi^{*}$ to $\vfi^{\diamond}_{L}$.
Substituting the results of equations~\ref{eq:overlap-penalty-case3} to~\ref{eq:case3-e2w-L} in~\ref{eq:branches-enclave-after-L-R} gives us:
  \begin{align*}
   LA(\vfi^{*}) - LA(\vfi^{\diamond}_{L}) \ge & \\
   & 2 \times (\Delta_1 + \beta_{i,2}) \\
   & - \beta_2 - \delta_{\vfi^{*}} (u_{2,1},-, V(B_{i,1})) \\
   & + \delta_{\vfi^{*}} (u_{1,j'},-, V(B_{i,2}))\\
   & + \delta_{\vfi^{*}} (u_{2,j''},-, V(B_{i,1}))\\
   & + \delta_{\vfi^{*}} (-,v_1, V(B_{i,2}))\\
   & -\delta_{\vfi^{*}} (v_2,-, V(B_{i,1}))
  \end{align*}
Considering the worst case scenario when $\Delta_1 = 0$, $\delta_{\vfi^{*}} (-,v_1, V(B_{i,2})) = 0$ and $\delta_{\vfi^{*}} (u_{1,j'},-, V(B_{i,2})) = 0$
\footnote{Namely the vertices of branch $B_{i,2}$ are labeled with a set of contentious integers and $v_1$ and $u_{1,j}$ are labeled before vertices of $B_{i,2}$ so that expands of $e_{1,w}$ and $e_{1,j'}$ stay unchanged.}, results in:
\begin{align*}
    LA(\vfi^{*}) - LA(\vfi^{\diamond}_{L}) \ge & \\
    & - (\delta_{\vfi^{*}} (u_{1,2},-, V(B_{i,1})) + \delta_{\vfi^{*}} (-,u_{2,1}, V(B_{i,2}))) + 2 \times \beta_{i,2} \ge \\
    & - \delta_{\vfi^{*}} (u_{1,2},-, V(B_{i,1})) +\beta_{i,2}
\end{align*}
Therefore $LA(\vfi^{*}) \le LA(\vfi^{\diamond}_{L})$ only if $\delta_{\vfi^{*}} (u_{2,1},-, V(B_{i,1})) \le \beta_{i,2}$. In this case, since $u_{1,2}$ has degree three with two outgoing edges from $E(C)$, and based on $\vfi^{\diamond}_{L}$, there is a path via edges of $E(C)$ going to the right most vertex and coming back. Due to this redundancy we can rearrange the vertices of $B_{i,1}$ in $\vfi^{\diamond}_{L}$, without increasing the value of linear arrangement, so that $u_{1,2}$ has the largest label among vertices $B_{i,1}$ (is the right most vertex of $B_{i,1}$). In this new layout $\vfi^{\diamond}_{L2}$ the length of edge $e_{1,2}$ will degrease by $\delta_{\vfi^{*}} (u_{1,2},-, V(B_{i,1}))$.

\noindent
Finally we have $LA(\vfi^{*}) - LA(\vfi^{\diamond}_{L2}) \ge \beta_{i,2}$, which contradicts the optimality of layout $\vfi^{*}$.
\paragraph{Case 3.3: $\alpha(e_{1,j'}) = -1$ and $\alpha(e_{2,j''}) = 1$.} In this case we suggest $LA(\vfi^{\diamond}_{R})$ as an alternative for $LA(\vfi^{*})$. Using the same approach and having the arithmetic details omitted, it can be verified that the following holds.
\begin{align*}
    LA(\vfi^{*}) - LA(\vfi^{\diamond}_{R}) \ge & \\
    & - (\delta_{\vfi^{*}} (-,u_{1,2}, V(B_{i,1})) + \delta_{\vfi^{*}} (u_{2,1},-, V(B_{i,2}))) + 2 \times \beta_{i,2} \ge \\
    & - \delta_{\vfi^{*}} (-,u_{1,2}, V(B_{i,1})) +\beta_{i,2}
\end{align*}
Using the same reasoning, $LA(\vfi^{\diamond}_{R})$ can be partially modified to have $u_{1,2}$ as the left most vertex of $B_{i,1}$ and reduce the length $e_{1,2}$ by $\delta_{\vfi^{*}} (-,u_{1,2}, V(B_{i,1}))$ and accordingly to have:
\[
    LA(\vfi^{*}) - LA(\vfi^{\diamond}_{R}) \ge \beta_{i,2} > 0
\]
\paragraph{Case 3.4: $\alpha(e_{1,j'}) = -1$ and $\alpha(e_{2,j''}) = 1$.} This case is symmetric to the case 3.3 and similarly it can be shown that $LA(\vfi^{*}) - LA(\vfi^{\diamond}_{L}) \ge \beta_{i,2} > 0$.

\end{Prxxx}
\begin{Prxxx}{(Proof of theorem~\ref{thrm:halin-OLA-extreme-vetrices})}
\label{proof:halin-OLA-extreme-vetrices}
    %This proof is implicitly constructed based on the results of
%lemmas~\ref{lemma:halin-layout-branch-one-side, lemma:halin-subtrees-labeled-separately, lemma:halin-branches-labeled-separately}.
Consider the case where for a given OLA $\vfi^{*}$ and extreme vertices $v$ and $u$, $d_T(v) \neq 1 \vee d_T(u) \neq 1$ or $d_T(v) \neq 1 \wedge d_T(u) \neq 1$.

\noindent
With no loss of generality we only present the case where $d_T(v) = k \ge 3$ and symmetrically it can be shown for $d_T(u)  \ge 3$ as well.
Hence there are $k-1$ spinal branches $B_{1,1}, \ldots, B_{1,k-1}$ connected to $v$. Based on lemma ~\ref{lemma:halin-subtrees-labeled-separately} and~\ref{lemma:halin-branches-labeled-separately},
$B_{1,1}, \ldots, B_{1,k-1}$ are separately labeled on the right side of $v$.
Each branch $B_{1,i}$ is anchored at vertex $v_i$ (is connected to $v$ via edge $\Set{v_i,v}$).
The set $\Set{B_{1,1}, \ldots, B_{1,k-1}}$ is connected to the rest of
the graph by exactly to edges $e_{1,j} \in E(C)$ and $e_{1,j'} \in E(C)$. Figure~\ref{fig:layout-of-left-extreme} generally presents layout $\vfi^{*}$.
Note that the two edges $e_{1,j}$ and $e_{1,j'}$ can not be initiated from the same branch. Assume $e_{1,j}$ and $e_{1,j'}$ are respectively
connected to $B_{1,j}$ and $B_{1,j'}$ 
and vertices of $B_{1,j}$ are labeled with integers smaller than the labels of vertices of $B_{1,j}$. 
Also notice that every two branches $B_{1,i}$ and $B_{1,i'}$ are connected
by at most one edge.
\begin{figure}
  \begin{center}
     \includegraphics [scale=.6] %[width=0.35\textwidth,height=5cm]%
          {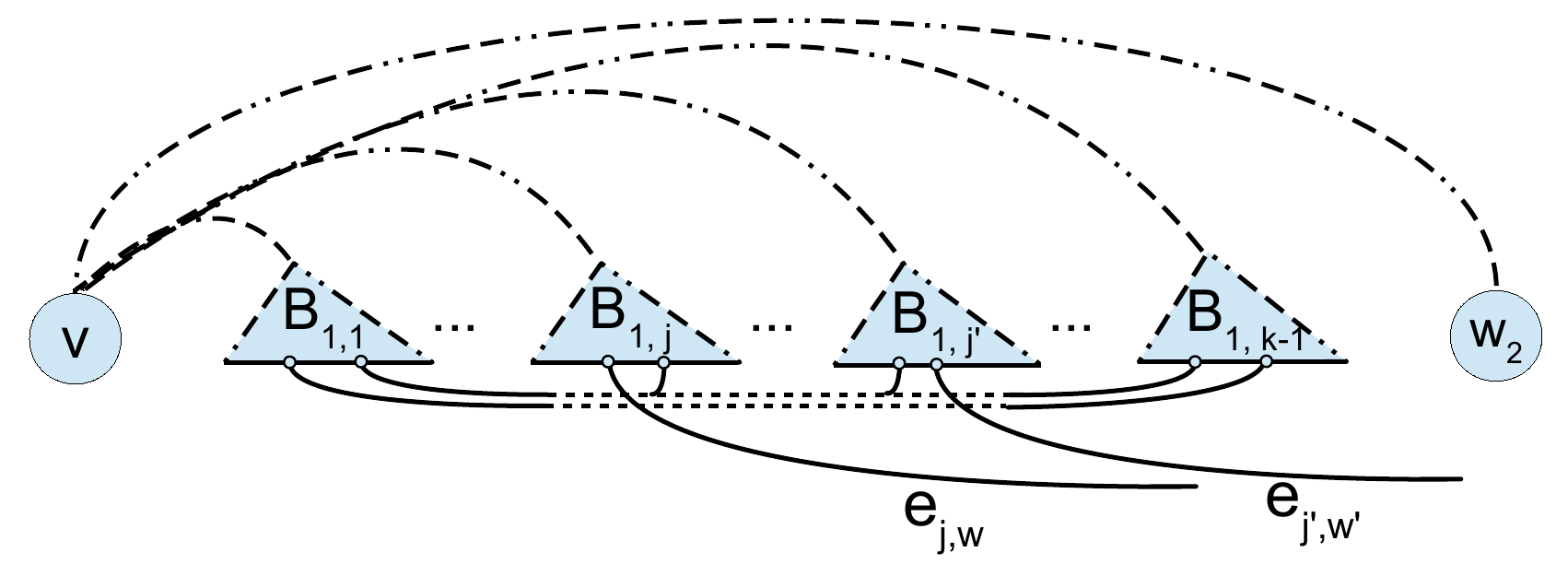}
  \end{center}
     \caption{General presentation of an OLA $\vfi^{*}$ for $H=T \uplus C$, where $\vfi^{*}(v)$ while $v$ is not a leaf in tree $T$.}
   \label{fig:layout-of-left-extreme}
\end{figure}
As apposed to layout $\vfi^{*}$ we present the following layout $\vfi^{\diamond}$.
\[
 \forall u \in V, \vfi^{\diamond} =
    \begin{cases}
       \vfi^{*}(u) + \sum_{j < i < k} \beta_{1,i} &\quad \text{if } u \in V(B_{1,j})\\
       \mathcal{T}_1 - \vfi^{*}(u) + 1 &\quad \text{if } u \in V(B_{1,j+1}) \cup \dots \cup V(B_{1,k-1})\\
       \mathcal{T}_1 - \vfi^{*}(u) - \beta_{1,j} + 1 &\quad \text{if } u \in V(B_{1,1})\cup \dots \cup V(B_{1,j-1}) \cup \Set{v} \\
       \vfi^{*}(u) &\quad \text{otherwise} \
     \end{cases}
\]
\noindent
Where $\beta_{i,j}$ and $\mathcal{T}_i$ are respectively the number of vertices of branch $B_{i,j}$ and subtree $T_i$\footnote{Hence $\mathcal{T}_1$ is equivalent to the largest possible label for the vertices of $T_i$}.

\noindent
Informally speaking, layout $\vfi^{\diamond}$ is constructed by mirroring the labels of vertices of all the branches about $v$, except for the vertices of  branch $B_{1,j}$.
Figure~\ref{fig:layout-of-left-extreme-after} depicts the layout $\vfi^{\diamond}$.
\begin{figure}
  \begin{center}
     \includegraphics [scale=.6] %[width=0.35\textwidth,height=5cm]%
          {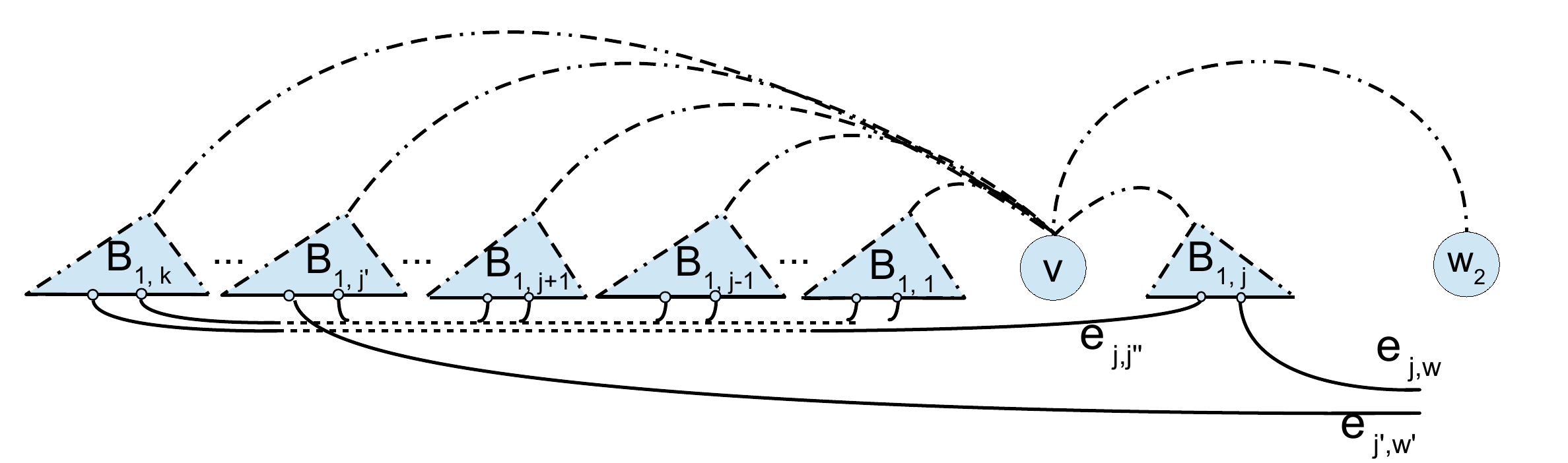}
  \end{center}
     \caption{General presentation of an OLA $\vfi^{\diamond}$ based on $\vfi^{*}$, where the label of the vertices of branches are mirrored about vertex $v$ except those of branch $B_{1,j}$.}
   \label{fig:layout-of-left-extreme-after}
\end{figure}
According to the construction of $\vfi^{\diamond}$ equation~\ref{eq:layout-of-left-extreme-after} holds.
Note that the relative inter-orders of vertices of $B_{1,i}$ for $1 \le i < k$ (accordingly the size of expands of internal edges) stay unchanged in $\vfi^{\diamond}$.
 \begin{align}
      \label{eq:layout-of-left-extreme-after}
    LA(\vfi^{*}) - LA(\vfi^{\diamond}) \ge & \\ \nonumber
    & (\lambda(e_{j,w}, \vfi^{*}) - \lambda(e_{j,w}, \vfi^{\diamond}))\ + \\ \nonumber
    & (\lambda(e_{j',w'}, \vfi^{*}) - \lambda(e_{j',w'}, \vfi^{\diamond}))\ + \\ \nonumber
    & (\lambda(e_{j,j''}, \vfi^{*}) - \lambda(e_{j,j''}, \vfi^{\diamond}))\ + \\  \nonumber
    & (\lambda(\Set{v,w_2}, \vfi^{*}) - \lambda(\Set{v,w_2}, \vfi^{\diamond}))\ +\\ \nonumber
    & (\lambda(\Set{v,v_j}, \vfi^{*}) - \lambda(\Set{v,v_j}, \vfi^{\diamond}))\ +\\ \nonumber
    & \sum_{j < i < k} (\lambda(\Set{v,v_i}, \vfi^{*}) - \lambda(\Set{v,v_i}, \vfi^{\diamond})) \nonumber
 \end{align}
Each term of equation~\ref{eq:layout-of-left-extreme-after} refers to the change in the expands of those edges that their expand may change
in the process of constructing $\vfi^{\diamond}$ from $\vfi^{*}$.

\noindent
It's not hard to verify that the following equations hold.
 \begin{align}
    & (\lambda(e_{j,w}, \vfi^{*}) - \lambda(e_{j,w}, \vfi^{\diamond})) =  \sum_{j < i <k} \beta_{1,i} \\
    & (\lambda(e_{j',w'}, \vfi^{*}) - \lambda(e_{j',w'}, \vfi^{\diamond})) \ge - \sum_{1 \le i <k} \beta_{1,i}\\
    & (\lambda(e_{j,j''}, \vfi^{*}) - \lambda(e_{j,j''}, \vfi^{\diamond})) \ge - \sum_{1 \le i \le j} \beta_{1,i} \\
    & (\lambda(\Set{v,w_2}, \vfi^{*}) - \lambda(\Set{v,w_2}, \vfi^{\diamond})) = \sum_{1 \le i < k,\ i \neq j} \beta_{1,i}\\
    & (\lambda(\Set{v,v_j}, \vfi^{*}) - \lambda(\Set{v,v_j}, \vfi^{\diamond})) = \sum_{1 \le i < j} \beta_{1,i}\\
    & \sum_{j < i < k} (\lambda(\Set{v,v_i}, \vfi^{*}) - \lambda(\Set{v,v_i}, \vfi^{\diamond})) = (k-j-1) \times \beta_{1,j}
 \end{align}
\noindent
Consequently equation~\ref{eq:layout-of-left-extreme-after} can be simplified as
$LA(\vfi^{*}) - LA(\vfi^{\diamond}) \ge (k-j-3) \times \beta_{1,j} + \sum_{j < i <k} \beta_{1,i}$.
The quantity $k-j-1$ is the number of branches labeled after $B_{1,j}$ and obviously $(k-j-1) \ge 1$. Hence for $(k-j-1) > 1$ or $(k-j-1) = 1 \wedge \beta_{j+1} > \beta_j$, we have $LA(\vfi^{*}) - LA(\vfi^{\diamond}) > 0$, which contradicts the optimality of $\vfi^{*}$. Now we analyze the case where
$(k-j-1) = 1 \wedge \beta_{j+1} \le \beta_j$.

\noindent
\paragraph{Case 1: $B_{1,j}$ and $B_{1,j+1}$ are not connected.} Referring to the structure of Halin graphs, this case holds only if $j > 1$. Informally speaking, there are some branches $B_{1,1}, \ldots, B_{1,j-1}$ which based on $\vfi^{*}$ their vertices are labeled after $v$ and before $B_{1,j}$. With respect to this case we construct a new layout $\vfi^{\circledast}$ where the labels of vertices of $B_{1,1}, \ldots, B_{1,j-1}$ are mirrored about $v$. Formally $\vfi^{\circledast}$ is constructed as:
\[
 \forall u \in V, \vfi^{\circledast} (u)=
    \begin{cases}
       \mathcal{B} + 2 - \vfi^{*}(u) &\quad \text{if } u \in V(B_{1,1}) \cup \ldots \cup V(B_{1,j-1}) \cup \Set{v}\\
       \vfi^{*}(u) &\quad \text{otherwise} \
     \end{cases}
\]
\noindent
Where $\mathcal{B} = \sum_{1 \le i < j} \beta_{1,i}$ is the number of vertices in set $\Set{V(B_{1,1}) \cup \ldots \cup V(B_{1,j-1})}$.

\noindent
Following the same approach as before we can show that $LA(\vfi^{*}) - LA(\vfi^{\circledast}) \ge \sum_{1 \le i < j} \beta_{1,i} > 0$ . The details of arithmetic calculations are left to the reader.

\noindent
\paragraph{Case 2: $B_{1,j}$ and $B_{1,j+1}$ are connected.} Hence $B_{1,j}$ and $B_{1,j+1}$ are the only branches connected to $v$. Figure~\ref{fig:layout-extreme-two-branches-left} shows the layout $\vfi^{*}$ corresponding to this case. As schematically shown in figure~\ref{fig:layout-extreme-two-branches-left-after}, we present the the alternative layout $\vfi^{\circledast}$, formally defined as it follows.
\[
 \forall u \in V, \vfi^{\circledast} (u) =
    \begin{cases}
       \beta_{1,1} + 2 - \vfi^{*}(u) &\quad \text{if } u \in V(B_{1,1}) \cup \Set{v}\\
       \vfi^{*}(u) &\quad \text{otherwise} \
     \end{cases}
\]
\begin{figure}
        \centering
        \begin{subfigure}[b]{0.3\textwidth}
                \includegraphics[scale=.6]
                {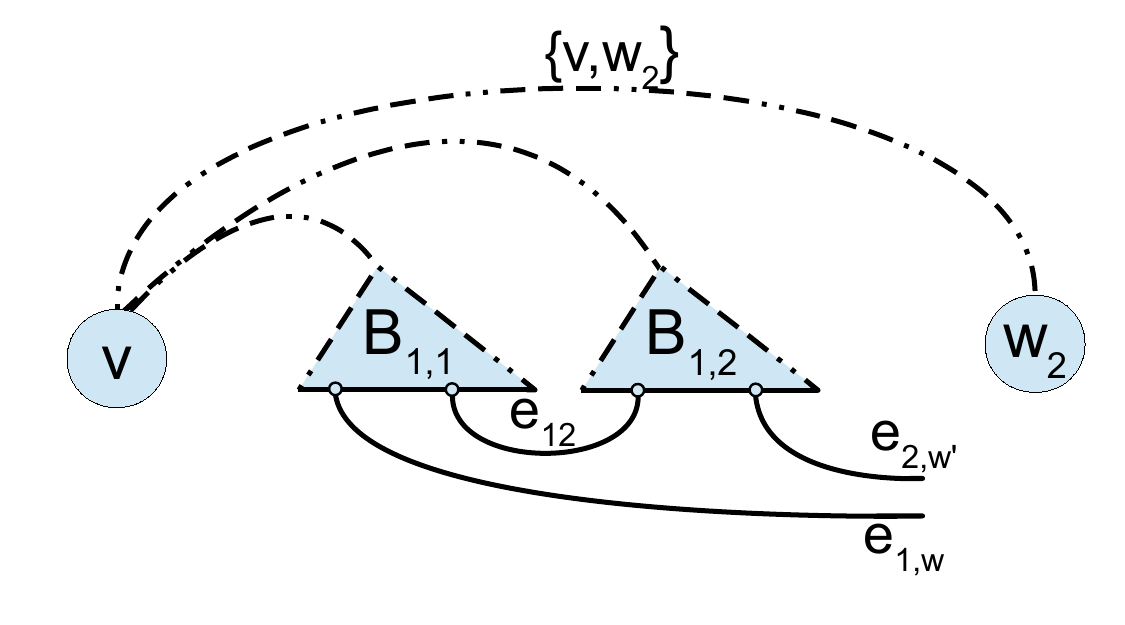}
                \caption{OLA $\vfi^{*}$ where $\vfi^{*}(v) = 1$  and $v$ is connected to exactly two branches $B_{1,1}$ and $B_{1,2}$.}
                \label{fig:layout-extreme-two-branches-left}
        \end{subfigure}
        \qquad\qquad\qquad
        \begin{subfigure}[b]{0.3\textwidth}
                \includegraphics[scale=.6]
                {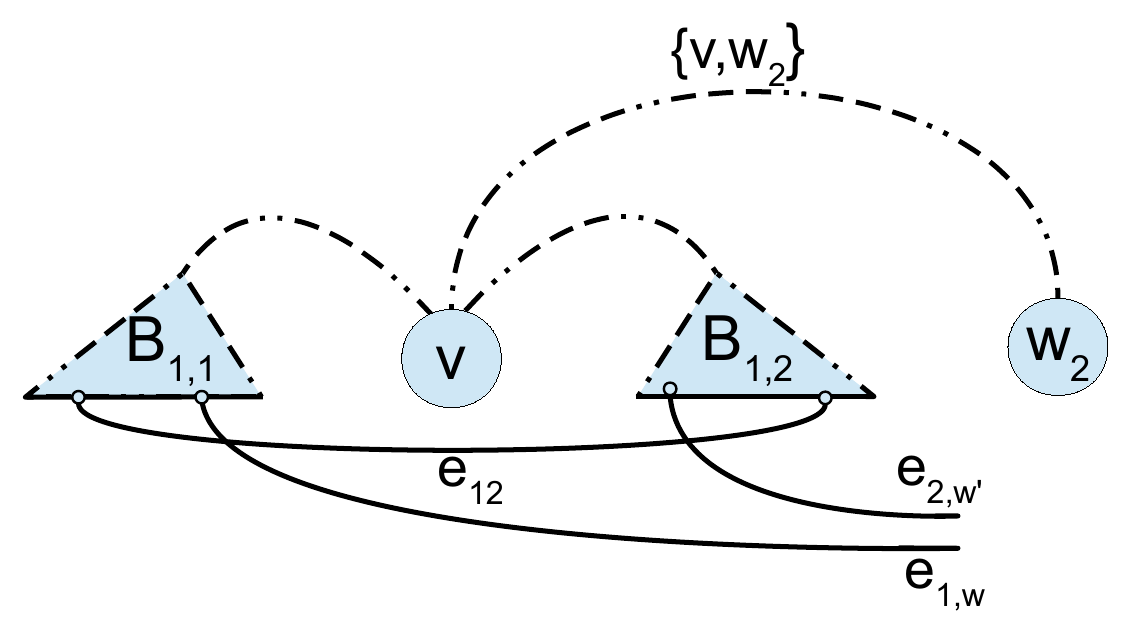}
                \caption{The alternative non-overlapping layout $\vfi^{\circledast}$ based on $\vfi^{*}$, where the labels of vertices in $B_{1,1}$ are mirrored about $v$.}
                \label{fig:layout-extreme-two-branches-left-after}
        \end{subfigure}
        \caption{OLA $\vfi^{*}$  and the corresponding alternative layout $\vfi^{\circledast}$.}
        \label{fig:layout-extreme-two-branches-left-and-after}
\end{figure}
Equation~\ref{eq:layout-extreme-two-branches-left} compares the value of linear arrangements $\vfi^{*}$ and $\vfi^{\circledast}$.
 \begin{align}
      \label{eq:layout-extreme-two-branches-left}
    LA(\vfi^{*}) - LA(\vfi^{\circledast}) = & \\ \nonumber
    & (\lambda(\Set{v,w_2}, \vfi^{*}) - \lambda(\Set{v,w_2}, \vfi^{\circledast}))\ + \\ \nonumber
    & (\lambda(\Set{v,v_2}, \vfi^{*}) - \lambda(\Set{v,v_2}, \vfi^{\circledast}))\ +\\ \nonumber
    & (\lambda(e_{1,2}, \vfi^{*}) - \lambda(e_{1,2}, \vfi^{\circledast}))\ + \\ \nonumber
    & (\lambda(e_{1,w}, \vfi^{*}) - \lambda(e_{1,w}, \vfi^{\circledast}))  \nonumber
 \end{align}
\noindent
Remember that $v_1$ and $v_2$ are the two vertices where $B_{1,1}$ and $B_{1,2}$ are anchored at.

\noindent
Based on the construction of $\vfi^{\circledast}$ from $\vfi^{*}$ we have:
 \begin{align}
    & \lambda(\Set{v,w_2}, \vfi^{*}) - \lambda(\Set{v,w_2}, \vfi^{\circledast}) = \beta_{1,1} \\
    & \lambda(\Set{v,v_2}, \vfi^{*}) - \lambda(\Set{v,v_2}, \vfi^{\circledast}) = \beta_{1,1} \\
    & \lambda(e_{1,2}, \vfi^{*}) - \lambda(e_{1,2}, \vfi^{\circledast}) \ge - \beta_{1,1} \label{eq:layout-extreme-two-branches-left-after-3} \\
    & \lambda(e_{1,w}, \vfi^{*}) - \lambda(e_{1,w}, \vfi^{\circledast}) \ge - \beta_{1,1} \label{eq:layout-extreme-two-branches-left-after-4}
 \end{align}

 \noindent
Finally putting equations~\ref{eq:layout-extreme-two-branches-left} to~\ref{eq:layout-extreme-two-branches-left-after-4} together
we conclude that $LA(\vfi^{*}) - LA(\vfi^{\circledast}) \ge 0$. But the equalities in equations~\ref{eq:layout-extreme-two-branches-left-after-3} and~\ref{eq:layout-extreme-two-branches-left-after-4} hold at the same time (and consequently $LA(\vfi^{*}) - LA(\vfi^{\circledast}) = 0$), only if the two edges $e_{1,2}$ and $e_{1,w}$ coincide at the left most vertex of
$B_{1,1}$. This situation in a Halin graph can only happen when branch $B_{1,1}$ has exactly one vertex.
For that reason we conclude that in an OLA for a Halin graph $H = T \uplus C$, a non-leaf vertex $v$ can be a extreme vertex, only if $v$ has exactly two leaves of $T$ as it's children\footnote{Remember that $\beta_{1,1} \ge \beta_{1,2}$.}.

\end{Prxxx}

\ifTR
%  \clearpage
\else
\fi

\end{document}